\documentclass[11pt]{article}

\usepackage{url}
\usepackage{amsmath,amsfonts,amsthm,amssymb,multirow}
\DeclareMathOperator*{\argmin}{arg\,min}
\usepackage{graphicx}
\usepackage{floatpag}
\usepackage{float}
\usepackage{times}
\usepackage{fullpage}
\usepackage{bbold}
\usepackage[pdfencoding=auto]{hyperref}
\usepackage{tikz}
\usetikzlibrary{cd}
\usepackage{comment}
\usepackage{calc}
\usepackage{upgreek}
\usepackage{mathtools}
\usepackage{color}
\usepackage{mleftright}
\allowdisplaybreaks
\newenvironment{reminder}[1]{\bigskip
	\noindent {\bf Reminder of #1  }\em}{\smallskip}

\DeclarePairedDelimiter\abs{\lvert}{\rvert}

\DeclarePairedDelimiter\norm{\lvert\lvert}{\rvert\rvert}

\newcommand{\normm}[1]{\left\lVert#1\right\rVert}
\newcommand{\abss}[1]{\left\lvert#1\right\rvert}

\DeclarePairedDelimiter\iprod{\langle}{\rangle}

\def \NP {\text{NP}}
\def \ka {{\kappa}}
\def \R {{\mathbb R}}

\def \Q {{\mathbb Q}}
\def \Z {{\mathbb Z}}
\def \poly {\text{poly}}
\def \plog {\text{polylog}}
\def \eps {{\varepsilon}}
\def \nnz {\text{nnz}}
\def \rk {\text{rank}}
\def \colspace {\text{colspace}}

\def \N {\mathcal{N}}

\def \tO {\Tilde{O}}
\def \tOm {\Tilde{\Omega}}
\def \bc {\mathcal{B}}

\newtheorem{theorem}{Theorem}[section]

\newtheorem{corollary}[theorem]{Corollary}

\newtheorem{definition}[theorem]{Definition}
\newtheorem{lemma}[theorem]{Lemma}
\newtheorem{claim}{Claim}
\newtheorem{conjecture}{Conjecture}

\makeatletter
\let\c@fconjecture\c@conjecture
\makeatother

\makeatletter
\let\c@fconj\c@conj
\makeatother

\def\tO{\tilde{O}}

\title{Optimal Fine-grained Hardness of Approximation of Linear Equations}

\author{Mitali Bafna\footnote{\texttt{mitalibafna@g.harvard.edu}} \\Harvard University \and Nikhil Vyas\footnote{\texttt{nikhilv@mit.edu}, Supported by NSF CCF-1909429.}\\MIT}

\begin{document}

\maketitle

\begin{abstract}
The problem of solving linear systems is one of the most fundamental problems in computer science, where given a satisfiable linear system $(A,b)$, for $A \in \R^{n \times n}$ and $b \in \R^n$, we wish to find a vector $x \in \R^n$ such that $Ax = b$.  The current best algorithms for solving dense linear systems reduce the problem to matrix multiplication, and run in time $O(n^{\omega})$. We consider the problem of finding $\eps$-approximate solutions to linear systems with respect to the $L_2$-norm, that is, given a satisfiable linear system $(A \in \R^{n \times n}, b \in \R^n)$, find an $x \in \R^n$ such that $\norm{Ax - b}_2 \leq \eps\norm{b}_2$. Our main result is a fine-grained reduction from computing the rank of a matrix to finding $\eps$-approximate solutions to linear systems. In particular, if the best known $O(n^\omega)$ time algorithm for computing the rank of $n \times O(n)$ matrices is optimal (which we conjecture is true), then finding an $\eps$-approximate solution to a dense linear system also requires $\tOm(n^{\omega})$ time, even for $\eps$ as large as $(1 - 1/\poly(n))$. We also prove (under some modified conjectures for the rank-finding problem) optimal hardness of approximation for sparse linear systems, linear systems over positive semidefinite matrices, well-conditioned linear systems, and approximately solving linear systems with respect to the $L_p$-norm, for $p \geq 1$. At the heart of our results is a novel reduction from the rank problem to a decision version of the approximate linear systems problem. This reduction preserves properties such as matrix sparsity and bit complexity.
\end{abstract}

\section{Introduction}

Algorithms for solving linear equations are one of the most fundamental primitives in computer science. Formally this is the problem where, given a linear system $(A, b)$, where $A \in \R^{m \times n}$ is a real matrix and $b \in \R^m$ is a vector in the column space of $A$, we need to find a vector $x \in \R^n$ such that $Ax = b$. Gaussian elimination running in time\footnote{Here we are discussing algorithms and hardness over the Real RAM, unless stated otherwise. We discuss the Word RAM in more detail in Section~\ref{word-ram-intro}.} $O(n^3)$ was one of the first algorithms for this problem. Hopcraft and Bunch~\cite{bunch} reduced solving linear equations to fast matrix multiplication of two $n \times n$ matrices~\cite{strassen, coppersmith, stothers, williams, Gall14a, josh-omega} which can be done in $m(n) = n^\omega$, where $\omega$ is the matrix multiplication constant. The current best known upper bound on $\omega$ is approximately $2.372..$~\cite{josh-omega}. The best known algorithms for solving linear systems reduce the problem to matrix multiplication, but there is no known reduction in the other direction. We study the complexity of finding approximate solutions to linear systems (defined more precisely later) under the following conjecture:

\begin{conjecture}[Rank-Finding Conjecture over RealRAM]  \label{conj:rk-dense} Finding the rank of a matrix $A \in \R^{m \times n}$ with $m = O(n)$ in RealRAM is $\Omega(n^\omega)$-hard.
\end{conjecture}

The problem of finding the rank of a matrix is a central problem in linear algebra. It is known that this problem can be reduced to fast matrix multiplication~\cite{bunch,IbarraMH82}, hence there exist algorithms for the rank-finding algorithm that run in time $O(n^\omega)$. There are known faster algorithms for restricted classes of matrices. For sparse matrices we know of $O(n^2)$ algorithms~\cite{Wiedemann86} and for low-rank matrices the $O(n^2+\rk(A)^\omega)$-time algorithms of~\cite{KaltofenS91, Eberly2, CheungKL13} run in time $n^{\omega-\Omega(1)}$ when $\rk(A) = n^{1-\Omega(1)}$. No improvement over the $O(n^\omega)$-runtime is known for general matrices though. We conjecture that this run time is in fact \emph{optimal} for general matrices. The rank-finding problem is equivalent to checking whether the determinant of a matrix is $0$. We get some evidence towards the truth of our conjecture by considering the computational model of arithmetic circuits\footnote{Such a reduction is unknown in the RealRAM model.}: In a seminal work Baur and Strassen~\cite{BS83} linearly reduced the the problem of matrix multiplication to the problem of computing the determinant, thus showing that the latter problem requires arithmetic circuits of size as large as those required for matrix multiplication. Furthermore, this is a central conjecture because falsifying it (getting faster algorithms for the rank-finding problem) would yield better algorithms for important problems like finding the size of a maximum matching in a graph~\cite{Lovasz79, MuchaS06}. \footnote{This is because maximum matching algorithms has a randomized reduction to the rank-finding problem.} For some direct evidence: there has been a line of work by Musco et al~\cite{MNSUW18} that gives algorithms to approximate the Schatten $p$-norms in time better that $O(n^{\omega})$ when $p > 0$. But at $p = 0$, the problem of finding the Schatten $p$-norm is the same as finding the rank of the matrix, and their algorithms run in time $\Omega(n^{\omega})$. Hence they are not able to beat the runtime of $O(n^\omega)$ to approximate the rank of a matrix, let alone determine it exactly.

Conjecture~\ref{conj:rk-dense} allows us to study the hardness of linear system solving and related linear algebraic problems in the style of fine-grained complexity~\cite{WW}. The conjecture implies $\Omega(n^{\omega})$-hardness of finding exact solutions to linear systems (Lemma~\ref{lem:eq}). One could hope to get faster algorithms though when allowed to find an approximate solution to the linear system. Specifically, we consider the following notion of approximation:

\begin{definition}[$\eps(n)$-Approximate Linear Search]\label{def:e-als}
For a function $\eps: \mathbb{N} \rightarrow [0,1]$, the $\eps$-Approximate Linear Search problem is defined as, given a satisfiable\footnote{Keeping with the convention of promise problems, we will assume that when given an unsatisfiable instance the algorithm is allowed to output an arbitrary vector.} linear system $(A\in \R^{O(n) \times n}, b)$, find an $x \in \R^n$ such that $\norm{Ax-b}_2 \leq \eps(n)\norm{b}_2$. 
\end{definition}

Note that the all $0$'s vector $x = 0^n$ is a $1$-approximate solution to any linear system as $\norm{A0^n-b} = \norm{b}_2$. Our main result shows that doing barely better than the trivial approximation is hard: $(1-1/n^{100})$-Approximate Linear Search i.e. finding an $x$ such that $\norm{Ax-b}_2 \leq (1-1/n^{100})\norm{b}_2$. is $\tOm(n^\omega)$ hard under Conjecture~\ref{conj:rk-dense}. We discuss our notion of approximation in Section~\ref{sec:approx-discuss} and show that showing hardness for this notion implies hardness for other notions of approximation that have been considered in the literature.

\noindent In a seminal work, Spielman and Teng~\cite{SpielmanT14} gave nearly linear-time algorithms ($O(n^2\log(1/\eps(n)))$-time) for finding $\eps(n)$-approximate solutions to Laplacian linear-systems and this result was built upon by many works~\cite{KyngLPSS16,CohenKKPPRS18} to give such algorithms for other restricted classes of linear systems. Our result shows that under the hardness of the rank-finding problem, these algorithms cannot be extended to general linear systems. As mentioned above, we conditionally rule out $\tilde{O}(n^2)$-time algorithms for finding $\eps(n)$-approximate solutions to general linear systems even for $\eps(n) = 1-1/n^{100}$. 

We also extend our results to give optimal conditional hardness of approximation (under analogous conjectures for the rank-finding problem) for restricted classes of linear systems: sparse linear systems, linear systems over positive-semidefinite matrices, and well-conditioned linear systems. We also show hardness of finding approximate solutions to linear systems with respect to the $L_p$ norm, for all $p \geq 1$.

Recently there has been a lot of progress in relating the exact time-complexities of various problems that have polynomial running times. Although there has been success in a variety of graph-theoretic, geometric and string problems~\cite{RodittyW13, AlmanW15, BackursI18}, there are very few fine grained reductions from the assumptions therein to linear algebraic problems, an example being, the work of Musco et al~\cite{MNSUW18} that  showed conditional lower bounds for spectrum approximation. We note here that there is a barrier to study the hardness of problems such as the rank-finding problem and solving linear equations under standard fine-grained complexity assumptions like SETH because these problems have linear time co-nondeterministic algorithms\footnote{For example, for an unsatisfiable linear system $(A, b)$ one can certify unsatisfiability by a witness $w$ such that $w^TA = \mathbf{0}$ but $w^Tb \neq 0$.} and SETH-based hardness would falsify the NSETH conjecture~\cite{CarmosinoGIMPS16}.

The theory of probabilistically checkable proofs~\cite{ALMSS,Hastad01} was instrumental in proving a host of $\NP$-hardness of approximation results. Though this theory was very successful in settling the time-complexity for approximation problems in NP, there are inherent limitations to extend these techniques to problems in P. Towards this, there has been recent progress for establishing hardness of approximation results for problems in P~\cite{ARW17, CGLRR19, SLM18}. Our paper makes further progress in this direction.

\subsection{Our results}\label{sec:results}
The table below gives a summary of our hardness results over the RealRAM. Note that all the below results can be extended to the WordRAM model of computation with some modifications (Section~\ref{word-ram-intro}). 

\renewcommand{\arraystretch}{1.5}
\begin{table}[H]
\begin{tabular}{|l|l|l|l|}
\hline
Model                    & Problem                  & Hardness & Algorithm         \\ \hline
\multirow{5}{*}{RealRAM} & $\eps$-ALS                  &  $\tOm(n^\omega)$ (Corollary~\ref{cor:dense-restated})         &   $O(n^\omega)$ (\cite{bunch})                \\ \cline{2-4} 
                         & Sparse $\eps$-ALS           &    $\tOm(n^2)$ (Corollary~\ref{cor:sparse})      &    $O(n^2)$(\cite{hestenes})               \\ \cline{2-4} 
                         & Well-conditioned $\eps$-ALS (Definition~\ref{def:wc-e-als}) &    $\tOm(n^\omega)$ (Corollary~\ref{cor:dense_c})      & \multirow{3}{*}{$O(n^\omega)$ (\cite{bunch})} \\ \cline{2-3}
                         & $\eps$-ALS with PSD matrix  &    $\tOm(n^\omega)$ (Corollary~\ref{cor:dense_psd})       &                   \\ \cline{2-3}
                         & $\eps$-ALS with $L_p$ norm (Definition~\ref{def:lp-e-als})  &    $\tOm(n^\omega)$ (Corollary~\ref{cor:dense_p})       &                   \\ \hline
\end{tabular}
\caption{$\eps$-ALS refers to $\eps$-Approximate Linear Search Problem. Our hardness results are under different conjectures, see statements for more details. All hardness results are for $\eps = 1-1/n^{100}$, while all algorithms are for the much (apriori) harder exact search problem. For small values of condition number better algorithms are known~\cite{hestenes} but for our regime of either unbounded or $\poly(n)$ condition number the above algorithms are the best known.}
\end{table}

We consider the question: Can one get fast approximate linear system solvers for general linear systems that run in time $\tilde{o}(n^\omega)$? We answer this question in the negative, under Conjecture~\ref{conj:rk-dense}. 
In this section we discuss the hardness result for solving general linear systems approximately (first row of Table 1). We discuss the other results of the table in Section~\ref{sec:intro-extensions}.

We refer to the decision version of the $\eps$-approximate linear search problem as Approximate Linear Decision problem (formally defined in Definition~\ref{def:approx_lin}). We prove all our hardness results by showing a reduction from the rank-finding problem to the Approximate Linear Decision problem.  We now state the lemma that proves our main reduction:
\begin{theorem}[Main reduction: Informal]\label{thm:red-main}
There exists a randomized Turing reduction from the rank-finding problem on $A \in \mathbb{R}^{m \times n}$ to the $(1 - 1/n^{100})$-Approximate linear decision problem\footnote{The constant 100 here is arbitrary and in fact our reduction works for all constants.} on square matrices with sparsity $\tO(\nnz(A))$ and dimension $O(\max(m,n))$. The reduction runs in time $\tO(\nnz(A))$ and works with high probability.
\end{theorem}

\paragraph{Our reduction:}
At the heart of our results lies an ``exact to approximate'' reduction for \emph{deciding} the satisfiability of linear systems. For showing hardness of finding approximate solutions, we are able to use this philosophy of ``increasing the gap'' between the YES and NO instances. We consider the following natural decision analogue of the $\eps$-approximate search problem: 
\begin{definition}[$\eps(n)$-Approximate Linear Decision problem~\cite{KyngZ17}]\label{def:approx_lin}
For a function $\eps: \mathbb{N} \rightarrow [0,1]$, given a linear system $(A \in \R^{m \times n}, b \in \R^m)$, with $m = O(n)$ and the promise that it falls into one of the following two sets of instances:
\begin{enumerate}
\item YES instance: There exists an $x \in \R^n$ such that $Ax = b$.
\item NO instances: For all $x \in \R^n$, $\norm{Ax-b}_2 > \eps(n) \norm{b}_2$,
\end{enumerate}
decide whether $(A,b)$ is a YES instance or a NO instance. We will refer to $\eps(n)$ as the ``gap'' of the instance.
\end{definition}

We show that the rank-finding problem reduces to the $(1-1/n^{100})$-approximate linear decision problem described above. We also show that the rank-finding problem is equivalent to the \emph{exact} linear decision problem i.e. the problem of deciding satisfiability of linear systems. Hence our reduction can be interpreted as increasing the gap between the YES/NO cases from almost 0 (can be arbitrarily small as we are working over the RealRAM) to $1-1/n^{100}$. We increase this gap in two stages, first to $\eps(n) = 1/n^{O(1)}$ and then to $1-1/n^{100}$. This gives conditional $\tOm(n^\omega)$-hardness of the $(1-1/n^{100})$-approximate linear decision problem. 

Even though the main reduction discussed here is from the rank-finding problem, we are also able to give a search to search reduction from the $1/n^{O(1)}$-approximate linear search problem to the $(1 - 1/n^{100})$-approximate linear search (see Corollary~\ref{cor:search-amp}).

\bigskip

We will now discuss the corollaries of the main reduction outlined above. Given Theorem~\ref{thm:red-main}, we perform a standard decision to search reduction (Lemma~\ref{lem:dec_to_search-r}) to get optimal hardness of approximation for the $(1-1/n^{100})$-approximate \emph{search} problem, under Conjecture~\ref{conj:rk-dense}. Thus under the rank finding conjecture, this reduction rules out all $(1-1/n^{100})$-approximation algorithms that run in time $\tilde{o}(n^{\omega})$.

\begin{corollary}[Informal]\label{cor:dense}
Under Conjecture~\ref{conj:rk-dense}, for all constants $c > 0$, the  $(1-1/n^{100})$-approximate linear search problem is $\tOm(n^\omega)$-hard in the RealRAM model of computation.
\end{corollary}

The second step of our reduction can also be used to increase the gap of the $1/n^{O(1)}$-Approximate Linear \emph{Search} problem from $1/n^{O(1)}$ to $1 - 1/n^{100}$. This gives us the following corollary:

\begin{corollary}[Search to Search reduction: Informal]\label{cor:search-amp}
If for any constant $a$ there exists an $\tO(n^{a})$-time algorithm for $(1-1/n^{100})$-approximate linear search problem then there exists an $\tO(n^{a})$-time algorithm for the $1/n^{100}$-approximate linear search problem.
\end{corollary}

Our hardness result is tight because there exist algorithms which solve the $(1-1/n^{100})$-Approximate Linear Search problem in $O(n^\omega)$ over the RealRAM. In fact, one can solve the more general problem of linear regression, i.e. given a (possibly unsatisfiable) linear system $(A,b)$, find an $x$ that minimizes $\norm{Ax-b}_2$, in time $O(n^\omega)$ (see Section~\ref{sec:approx-discuss}).

\subsection{Extensions}\label{sec:intro-extensions}
We prove several extensions of our main theorem using the reduction discussed above in the appendix. We can modify our main reduction so that it preserves the sparsity and condition number of the original instance to get the results for sparse and well-conditioned linear systems. To get hardness for PSD linear systems we need additional ideas beyond this reduction. We get the following results:

\paragraph{Sparse Linear Systems:} Linear equation solving has also been studied in the case of sparse linear systems. We know of $O(\nnz(A)n)$ time algorithms for solving a linear system $(A \in \mathbb{R}^{O(n) \times n}, b)$ where $\nnz(A)$ denotes the number of non-zero entries of $A$, that use Conjugate Gradient Descent~\cite{hestenes}, so that when the sparsity of $A$ is $\tO(n)$, these algorithms run in time $\tO(n^2)$. Our reduction (discussed above), preserves the sparsity of the original matrix $A$ and thus reduces the exact problem over sparse linear systems to the approximate problem over sparse ones. We start with an analogous conjecture to Conjecture~\ref{conj:rk-dense} for finding the rank of a sparse matrix:
\begin{conjecture}[Rank-finding Conjecture for Sparse matrices over RealRAM]\label{conj:rk-sparse}
Finding the rank of a matrix $A \in \R^{m \times n}$ with $m = O(n)$ and $\nnz(A) = \tO(n)$ in RealRAM is $\Omega(n^2)$-hard.
\end{conjecture}

Under the above conjecture we show (see Corollary~\ref{cor:sparse}) that solving $(1-1/n^{100})$-Approximate Linear Search problems on sparse linear systems is $\tOm(n^2)$-hard, which is optimal up to poly-logarithmic factors.\\

\paragraph{Positive Semi-Definite Linear systems:} We give optimal hardness for the $(1-1/n^{100})$-approximate linear search problem when the matrix $A$ is restricted to be positive semidefinite (see Corollary~\ref{cor:dense_psd}). Recently there has been a lot of work for getting nearly linear-time approximation algorithms for restricted classes of matrices. For a slightly more restricted class of linear systems than PSD ones, called Strongly Diagonally Dominant (SDD) systems, Spielman and Teng gave near-linear time approximate solvers~\cite{SpielmanT14}, leaving near-linear time approximation algorithms for PSD linear systems as the next open problem. In fact, resolving the time-complexity for PSD linear systems was in mentioned as an open problem in~\cite{AndoniKP19}, where they gave unconditional hardness for PSD linear systems for sublinear-time algorithms. Interestingly, we show that under Conjecture~\ref{conj:rk-mid}, such solvers are not possible for PSD linear systems, thus giving a conditional separation of the time complexity required for approximately solving SDD linear systems versus PSD ones.

\paragraph{Well-conditioned linear systems:} 
We now turn our attention to the problem of solving well-conditioned (see Definition~\ref{def:cond-num}) linear systems approximately. In Section~\ref{sec:approx-rank}, we give optimal conditional hardness of approximation for linear systems over matrices with polynomially bounded condition number under the $\tOm(n^\omega)$-hardness of the well-conditioned rank-finding problem. 

We also prove the following analogue of Corollary~\ref{cor:search-amp} which amplifies the gap for the \textit{search} problem on well-conditioned matrices:
\begin{corollary}\label{cor:search-amp_c}
If there exists as $\tO(n^{a})$ time algorithm for Well-conditioned $(1-1/n)$-Approximate Linear Search then there exists a $\tO(n^{a})$-time algorithm for Well-conditioned $(1/n)$-Approximate Linear Search problem.
\end{corollary}

\noindent \textbf{Finding $L_p$ approximate solutions:} In Section~\ref{sec:lp}, we also get similar optimal conditional hardness of approximation results (see Corollary~\ref{cor:dense_p}) for finding $L_p$-approximate solutions for a satisfiable linear system $(A, b)$, when $p \geq 1$ i.e. find $x \in \R^n$ that satisfies $\norm{Ax-b}_p \leq \eps(n)\norm{b}_p$. Amongst these norms, the $L_1$ and $L_\infty$ norms are particularly important. The $L_1$-norm corresponds to solving linear equations with minimum deviations, also known as the Least absolute deviations problem and the $L_{\infty}$-norm is the problem of minimizing error under the Chebyshev criterion. For all constants $c$, we give $\tOm(n^\omega)$-time hardness for solving linear systems $(1-1/n^{100})$-approximately under this notion, assuming the Conjecture~\ref{conj:rk-dense}. For sparse matrices, we can show $\tOm(n^2)$-time hardness for solving linear systems $(1-1/n^{100})$-approximately under this notion, assuming Conjecture~\ref{conj:rk-sparse}.

\subsection{Reductions over the WordRAM}\label{word-ram-intro}
Our main reduction can be modified (see Lemma~\ref{lem:red-main_w}) to preserve the bit-complexity of the original matrix. In Section~\ref{sec:l2_w}, we can show analogous results for all the problems considered above over the WordRAM. We show the conditional hardness result for general linear systems over the WordRAM in Section~\ref{sec:l2_w}. We show that under analogous conjectures for the rank-finding problem over the WordRAM (see Conjecture~\ref{conj:rk-dense-w}), the problem of finding $(1-1/n^{100})$-approximate solutions to linear systems with bit-complexity $O(\log n)$ is $\tOm(n^\omega)$-hard over the WordRAM (see Corollary~\ref{cor:dense_w}). Our hardness result is tight up to polylogarithmic factors because there exist $\tO(n^\omega)$-time algorithms for exactly solving linear systems on the WordRAM~\cite{Storjohann05,PS12,BirmpilisLS19}. 

We omit proofs of the other reductions over the WordRAM but they follow easily using ideas very similar to those outlined in Section~\ref{sec:l2_w}. For sparse matrices over WordRAM we are only able to give a $\tOm(n^{2})$ lower bound, which is trivial as we might need $\tOm(n^{2})$ bits to even describe a solution. On the algorithmic side, no improvement over the dense case algorithmic runtime of $O(n^{\omega})$ was known until very recently, when Peng and Vempala~\cite{PengV21} succeeded in finding an asymptotically faster algorithm for the $1/\poly(n)$-approximate linear search problem.

\subsection{Further applications and related work}
Recently, there has been a lot of progress on the algorithmic front for finding approximate solutions to restricted classes of linear systems $(A,b)$. In a breakthrough work, Spielman and Teng~\cite{SpielmanT14} obtained $\tilde{O}(\nnz(A)\log(1/\eps(n)))$-time algorithms for finding $\eps(n)$-approximate solution to  Laplacian systems and Strongly Diagonally Dominant (SDD) systems. This result was followed up by algorithms for more general classes of linear systems such as Connection Laplacians~\cite{KyngLPSS16} and Directed Laplacian systems~\cite{CohenKPPRSV17}. This raised the hope that such approximation algorithms could be obtained for more general classes of matrices such as truss stiffness matrices and total variation matrices. Kyng and Zhang~\cite{KyngZ17} showed that such algorithms for these slightly more general classes would imply approximation algorithms for general linear systems. Therefore, by composing our reduction with theirs, one immediate corollary we get is that solving approximately for these classes of restricted linear systems is as hard as the rank-finding problem. 

In~\cite{KWZ20} the authors prove conditional hardness for the problems of Packing/Covering Linear Programs based on the hardness of approximately solving general linear equations. Prior to our work there was no evidence of hardness for approximately solving linear equations. Our results therefore imply hardness for these problems under the rank-finding problem.

\paragraph{Organization:} In section~\ref{sec:prelims} we introduce notation and basic definitions that will be used throughout the paper. In Section~\ref{sec:main-red} we give a proof of our main reduction (Theorem~\ref{thm:red-main-formal}). In 
Section~\ref{sec:l2}, we show conditional hardness for finding approximate solutions to linear systems over the Real RAM. We then extend these conditional hardness results to restricted classes of linear systems: Section~\ref{sec:sparse_l2} considers sparse linear systems, Section~\ref{sec:psd} considers positive semidefinite linear systems, and finally Section~\ref{sec:approx-rank} considers well-conditioned linear systems. In Section~\ref{sec:lp}, we show the conditional hardness of finding approximate solutions to linear systems in the $L_p$-norm for $p \geq 1$. In Section~\ref{sec:l2_w} we show the analogues of these results over the WordRAM model of computation.

\section{Preliminaries}\label{sec:prelims}
Below is some notation that will be used throughout:

\paragraph{Notation:} We will use $\nnz(A)$ for to denote the sparsity of a matrix $A$ and we will assume that $\nnz(A) \geq \max(n, m)$ for $A \in \mathbb{R}^{m \times n}$. We will call a matrix $A \in \mathbb{R}^{m \times n}$ sparse if $\nnz(A) = \tO(\max(m, n))$. We will use $\kappa(A)$ to denote the condition number of a matrix $A$. By bit complexity of $A$ or $\bc(A)$ we will refer to maximum bit complexity of any entry in the matrix/vector $A$. We will use $\leq_T$ to denote Turing reductions. Most of our Turing reductions run in quasi-time linear in the input size. We say $\iprod{a, b}$ to denote the inner product of $a$ and $b$ i.e. $\sum_i a_ib_i$. We denote $W^{\perp}$ to denote the subspace orthogonal to the subspace $W$. By $P_W(b)$ we denote the projection of $b$ on the vector space $W$. For a matrix $M$ we denote its column space by $\colspace(M)$. By $\norm{v}_p$ we mean the $L_p$ norm of $v$, whenever $p$ is not specified we mean the $L_2$ norm. By $A^{\dagger}$ we mean the pseudoinverse of a matrix $A$. For a matrix $A \in \mathbb{R}^{m \times n}$ by $\Pi_A = A(AA^T)^{\dagger}A^{T}$ we mean the linear operator such that for all $x \in \mathbb{R}^m$, $\Pi_A(x)$ is the projection of $x$ on $\colspace(A)$. By $g = \tO(f)$ we mean $g = O(f \cdot \plog(f))$. By $g = \tOm(f)$ we mean $g = \Omega(f/\plog(f))$. Whenever not specified by algorithms we mean randomized algorithms. We use w.h.p. to denote a probability of $1 - 1/n^{\log n}$, where $n$ is the input size under consideration.

We will use the following definitions later on:

\begin{definition}[Turing reductions]
There exists a randomized Turing reduction from problem $A$ to problem $B$ that works with high probability when the following holds: There exists an algorithm which solves problem $A$ with high probability, given an oracle for problem $B$. Unless specified otherwise, we use whp to denote a probability of $1 - 1/n^{\log n}$, where $n$ is the input size of problem $A$. Note that any reduction that is correct with probability $\geq 2/3$ for decision problems $A$ and $B$, can be amplified to obtain a success probability of $1 - 1/n^{\log n}$, by running the original reduction $\poly\log(n)$ times and taking a majority vote.
\end{definition}

\begin{definition}[Matrix Multiplication Constant]
The matrix multiplication exponent $\omega$ is defined to be the smallest constant such that there exists an $O(n^{\omega})$-time algorithm for matrix multiplication. The current best bound on $\omega$ is  $2.372..$~\cite{josh-omega}.
\end{definition}

\begin{definition}[Condition number]\label{def:cond-num}
Condition number of a matrix $A \in \R^{m \times n}$ where $n \leq m$ and $A$ has full column-rank is defined as: 
$$\kappa(A) = \frac{\max\limits_{\norm{x}_2 = 1}\norm{Ax}}{\min\limits_{\norm{y}_2 = 1}\norm{Ay}}.$$

Note that when the matrix does not have full column rank the denominator becomes $0$ and hence the condition number is undefined, so whenever we discuss the condition-number we assume that we are working with matrices that have full column-rank. For a square invertible matrix $A$ we get that $\kappa(A) = \frac{\sigma_{1}}{\sigma_{n}}$, where $\sigma_i$ is the $i^{th}$ singular value.

We use the term well-conditioned to refer to a matrix with $\poly(n)$ condition-number.
\end{definition}

We will use many standard properties about Gaussians in Section~\ref{sec:l2} and we state them here without proof:
\begin{lemma}\label{lem:gaussian}
Let $x \sim \N(0, \sigma^2), b = (b_1, b_2, \ldots, b_m) \sim \N(0, \sigma^2)^m$ and $v = (v_1, v_2, \ldots, v_m)$:
\begin{enumerate}
    \item $\Pr[\abs{x} \leq a] \leq a/\sigma$\label{prop:gp1}
    \item $\Pr[\abs{x} \geq a] \leq \sigma^2/a^2$\label{prop:gp2}
    \item $\Pr[\norm{b}^2 \geq am\sigma^2] \leq 1/a$\label{prop:gp3}
    \item $\iprod{v, b}$ is distributed as $\N(0, \norm{v}^2\sigma^2)$\label{prop:gp5}
\end{enumerate}
\end{lemma}

We use the following lemma in Section~\ref{sec:l2}:
\begin{lemma}\label{lem:expander}
Let $G = (V, E)$ be a cycle on $n$ vertices with Laplacian $L$, $\abs{V} = n$. For all vectors $y$ with, such that for all $y$ satisfying $y \cdot \mathbf{1}^n = 0$:
$$y^T Ly = \sum_{(i, j) \in E} (y_i - y_j)^2 = \sum_{i \in [n]} (y_i - y_{i + 1})^2 \geq \norm{y}^2/n^2$$where addition wrt $i$ is modulo $n$.
\end{lemma}
\begin{proof}
    For a graph $G$, $y^T Ly = \sum_{(i, j) \in E} (y_i - y_j)^2 \geq \lambda_2\norm{y}^2$, where $\lambda_2$ is the second largest eigenvalue of the Laplacian matrix. As $\lambda_2 \geq 1/n^2$ for the $n$-cycle, we are done.
\end{proof}

The proof of the next lemma follows from standard linear algebra and hence we omit it.

\begin{lemma}\label{lem:solns}
Let $(A, b)$ be a linear system. Let $W = \colspace(A)^{\perp}$ and $P_W(b)$ denote the projection of the vector $b$ onto the subspace $W$. Then we have that,
$\min_x \norm{Ax-b}_2 \geq \norm{P_W(b)}_2$.
\end{lemma}

\section{Proof of Main Reduction (Theorem~\ref{thm:red-main})}\label{sec:main-red}
In this section we will prove the reduction from the rank-finding problem to the approximate version of the linear decision problem. For simplicity, throughout this section we work on the RealRAM model of computation and wherever we do not state it we assume that this is the case, so we \emph{do not} discuss the bit complexity of the reductions. Our reduction can be modified to work on the WordRAM which we do in Section~\ref{sec:l2_w}.

\begin{theorem}[Restatement of Theorem~\ref{thm:red-main}]\label{thm:red-main-formal}
For all constants $c > 0$, there exists a randomized Turing reduction in the RealRAM model of computation, from the rank-finding problem on $A \in \mathbb{R}^{m \times n}$ to the $(1 - 1/n^c)$-Approximate linear decision problem on $(A' \in \R^{n' \times n'} , \mathbf{1}^{n'})$, with dimension $n' = O(\max(m,n))$ and sparsity $\tO(\nnz(A))$, where in the YES case we have the additional property that the matrices produced have full rank. The reduction runs in time $\tO(\nnz(A))$, produces $\plog(mn)$ instances of the approximate linear decision problem and works with high probability.
\end{theorem}
We will prove this reduction in three steps:

\begin{enumerate}
    \item Lemma~\ref{lem:rank}: Rank-Finding Problem $\leq_T$ Full Rank problem (see Definition~\ref{def:fr})
    \item Lemma~\ref{lem:fc_to_lda}: Full Rank problem $\leq_T$ $(1/n^{O(1)})$-Approximate linear decision problem 
    \item Lemma~\ref{lem:amp}: $1/n^{O(1)}$-Approximate linear decision problem $\leq$ $(1 - 1/n^{c})$-Approximate linear decision problem 
\end{enumerate}

Given the lemmas, it will be straightforward to combine these reductions to get that the linear decision problem reduces to the $(1 - 1/n^{c})$-Approximate linear decision problem. 
Let us now show each of the steps stated above. We will prove the first reduction from the rank-finding problem to the full-rank problem. This is similar to a reduction by Wiedemann~\cite{Wiedemann86} for finite fields. Let us formally introduce the Full rank problem.

\begin{definition}[Full Rank Problem.]\label{def:fr} Given a matrix $A \in \R^{n \times n}$, is $\rk(A) = n$? \end{definition}

The intuition behind the following reduction from the rank-finding problem to the full-rank problem is the following: For a matrix $A \in \R^{m \times n}$ with $m > n$ and $\rk(A) = k < n$, adding $t$ ``random'' columns will give a full column rank matrix if and only if $t \geq n-k$. Hence we can binary search for the first value of $t$ which gives us a full column rank matrix, yielding $k$.

We start by showing that adding a random column to a (non full column rank) matrix increases the dimension of its column space by $1$ with high probability. To maintain the sparsity of the initial matrix, we add random sparse Gaussian vectors. This is similar to a reduction by Wiedemann~\cite{Wiedemann86} for finite fields. 

\begin{lemma}\label{lem:add_random} 
Let $M \in \R^{m \times n}$ be a matrix of rank $\geq r$. Then there exists a sampling procedure to sample $z \in \R^m$, with $\nnz(z) = \min(m, m\log^2(m)/(m-r +1))$, that runs in time $\tO(\nnz(z))$, such that the matrix $B = [M ~ z]$ has $\rk(B) \geq \min(r+1, m)$ whp. 
\end{lemma}
\begin{proof}
Let $s = \min(m, m\log^2(m)/(m-r+1))$. If $s = m$ let $S = [m]$, otherwise randomly choose $S \subseteq [m]$, by sampling $s$ coordinates from $[m]$ with replacement. If $i \notin S$, then set $z_{i} = 0$, else sample it from $\mathcal{N}(0,1)$. 

Now consider the matrix $B = [M ~ z]$. Suppose that $\rk(M) \geq \min(r+1, m)$, then we already have that $\rk(B) \geq \rk(M) \geq \min(r+1, m)$. So we need to only prove that when $\rk(M) = r < m$, then this procedure gives $\rk(B) \geq r+1$. This is equivalent to proving that $z \notin \colspace(M)$.

Let $W$ be the orthogonal subspace to $\colspace(M)$. $W$ has dimension $m-r > 0$ and for all vectors $w \in W, u \in \colspace(M)$, $\iprod{w,u} = 0$. We will prove that with large probability there exists a vector $v \in W$, such that $\iprod{v,z} > 0$, which implies that $z \notin \colspace(M)$. 

Since $\rk(W) = m - r$, by Lemma~\ref{lem:get_nz} there exists vector $v \in W$ which is non-zero on at least $m-r$ coordinates. Let $G$ denote the set of coordinates where $v$ is non-zero with $\abs{G} \geq m-r$.

Since $S$ (the set of non-zero coordinates of $z$) is a large enough random set of coordinates, one can prove that $S,G$ have a non-zero intersection. Formally,

\begin{align*}
    \Pr[S \cap G = \phi] &= \left(\Pr_{i \sim [m]}[i \not\in G]\right)^{\abs{S}}\\
    &\leq \left(r/m\right)^{m\log^2(m)/(m-r+1)}\\
    &\leq \frac{1}{m^{\log m}}.
\end{align*}

Let us now consider the inner product between $z,v$. We have that, $\iprod{z,v} = \sum_{i \in S \cap G} z_{i} v_i$, is a Gaussian random variable with non-zero variance if $S \cap G \neq \phi$, which implies that $\iprod{z,v} \neq 0$ whp. This implies that $z \notin \colspace(M)$, and the matrix $B = [M ~ z]$ has rank equal to $\min(r+1, m)$ whp.
\end{proof}

\begin{lemma}\label{lem:get_nz}
Let $W \subseteq \mathbb{R}^n$ be a vector space of dimension $d$. Then there exists a vector $v \in W$ such that $v$ has non-zero entries on at least $d$ coordinates. 
\end{lemma}
\begin{proof}
Since $\rk(W) = d$, there exists vectors $v_1,\ldots,v_j \in W$ such that $G = \{k \mid \exists i (v_i)_k \neq 0\}$ satisfies $\abs{G} \geq m - r$, else $W$ would be a subset of $\R^{< m-r}$. We can now take a random linear combination of these vectors to get a vector that is non-zero on every coordinate in $G$. Consider the vector $v = \sum_i x_i \cdot v_i$, where each $x_i \sim \N(0,1)$. Using property~\ref{prop:gp5} of Lemma~\ref{lem:gaussian}, we get that for every coordinate $j \in G$, $v_j$, is a Gaussian with zero mean and non-zero variance, hence the probability that it is zero is a measure-zero set. A union bound over all coordinates in $G$, gives us that whp, $v$ in nonzero on all coordinates in $G$. From now on we assume that this is the case (as otherwise we can consider the reduction to have failed).
\end{proof}

We will now prove the reduction from Rank-Finding problem to the Full Rank problem.

\begin{lemma}\label{lem:rank} 
There exists a randomized Turing reduction which works w.h.p. from the Rank-Finding problem on $A \in \R^{m \times n}$ to the Full rank problem, that runs in time $\tO(\nnz(A))$ and produces $O(\log m)$ instances of the Full rank problem, such that all instances of the matrices produced have dimension $O(\max(m,n)) \times O(\max(m,n))$ and sparsity $\tO(\nnz(A))$.
\end{lemma}

\begin{proof} Consider the linear system $Ax = b$. We will show that deciding whether there exists a satisfying solution reduces to the Full rank problem through the following sequence of reductions:

Linear System Decision $\leq_T$ Rank-Finding Problem $\leq_T$ Rank-Decision Problem $\leq_T$ Full Column Span $\leq_T$ Full Rank Problem 

We will show that each of these steps preserves the sparsity of the matrix. Let us now describe every reduction above:

\begin{enumerate}
    \item Suppose we want to check satisfiability of $Ax = b$. Consider the matrix $A' = [A ~~ b~]$, i.e. the matrix formed by appending the column vector $b$ to the matrix $A$. The system $Ax = b$ has a solution if and only if $rk(A) = rk(A')$, so it suffices to  find the rank of the matrices $A,A'$. Hence we have reduced the linear system decision problem to the Rank-Finding Problem. Also note that $\nnz(A') = s+m = O(\nnz(A))$ as $\nnz(A) \geq m$. 
    \item Suppose we want to find the rank of a matrix $M \in \R^{m \times n}$. One can see that the rank of $M$ is the largest integer $k \in [m]$ for which $\rk(M) \geq k$. Consider the Rank-Decision problem: Given a matrix $M$, is $\rk(M) \geq k$ or $< k$? Thus to find rank of $M$ one can do a binary search over $k \in [m]$, to find the largest $k$ for which $\rk(M) \geq k$, using $\log m$ calls to the rank-decision problem.
    \item Given a matrix $M \in \R^{m \times n}$, suppose we want to check if $\rk(M) \geq k$ or $< k$. From Lemma~\ref{lem:add_random}, we know that if $\rk(M) \geq k$, and we added a random vector $z_1$ sampled according to the procedure therein, the rank of the new matrix $B_1 = [M ~ z_1]$ will be $\geq \min(r+1, m)$ whp. The lemma gives us that the sparsity of $z_1$ is $\min(m, m\log^2m/(m-k+1))$. Similarly, consider the matrix $B_{m-k} = [M ~ z_1 ~ \ldots ~z_{m-k}]$, where each $z_i$ is sampled according to the sampling procedure in Lemma~\ref{lem:add_random} and $\nnz(z_i) = \min(m, m\log^2m/(m-k-i+2))$. One can prove by a union bound on $i$, that whp, if the rank of matrix $B$ was $\geq \min(k,m)$ then the rank of each intermediate matrix $B_i$ is $\geq \min(k+i,m)$ whp, and hence $\rk(B_{m-k}) \geq \min(k+m-k,m) = m$, that is, $B_{m-k}$ has full column span.
    
    On the other hand, if $\rk(B) < k$, then since we have appended only $m-k$ vectors in creating $B_{m-k}$, the rank of $B_{m-k} < m$. 
    
    Hence to check if $\rk(M) \geq k$ or not, it suffices to check if $\rk(B_{m-k}) = m$ or not. Note that the sparsity of $B_{m-k}$ is $\nnz(M) + \sum_i \nnz(z_i) \leq \nnz(M) + m\log^2m\sum_{i \in [m]} 1/i = \nnz(M) + \tO(m) = \tO(\nnz(M))$.
    
    \item Given a matrix $M \in \R^{m \times n}$, suppose we want to check if $\rk(M) = m$ or $< m$, i.e. does it have Full column span or not. The problem is trivial if $m > n$ as then the rank of the matrix is $\leq n$ and the instance is a NO instance. So we will assume that $m \leq n$. In this case, similar to the reduction above we will add rows $z_1, z_2, \ldots, z_{m-n}$. One can prove that this preserves sparsity and creates a full rank matrix whp if $\rk(M) = m$ to begin with.
\end{enumerate}

Since each step succeeds whp, we can take a union bound over all steps to say that the whole reduction succeeds whp. Then we can compose all the reductions to get that the Linear system decision problem reduces to $O(\poly\log m)$ instances of the Full rank problem.
\end{proof}

We will now prove that the full rank problem reduces to the $(1/n^{O(1)})$-Approximate linear decision problem. As we will see in section~\ref{sec:eq} the full-rank problem is equivalent to the linear decision problem. Hence we can think of the next step as reducing the linear decision problem to the $(1/n^{O(1)})$-Approximate linear decision problem i.e. increasing an arbitrarily small gap between the YES and NO cases to an inverse-polynomial gap.

The idea behind the proof is that for a full-rank square matrix $A \in \mathbb{R}^{n \times n}$ every vector $b \in \mathbb{R}^{n}$ belongs in the column space of $A$ and hence the linear system $(A, b)$ is satisfiable. 
On the other hand if $\rk(A) < n$ we expect a random vector $b$ to be outside the column space of $A$ and hence we expect the linear system $(A, b)$ to be unsatisfiable. 
We show a strengthened version of the previous statement by proving that w.h.p. for a random Gaussian vector $b$, it holds that for all $x$, 
$\norm{Ax-b} \geq \eps(n)\norm{b}$ 
for some $\eps(n) = 1/n^{O(1)}$. 
We also show that we can rescale the rows of the linear system $(A, b)$ to get the system $(A', \mathbf{1}^n)$ while maintaining the property that $\norm{A'x-\mathbf{1}^n} \geq \eps(n)\norm{\mathbf{1}^n}$. We perform the rescaling to obtain $b = \mathbf{1}^n$ as that will simplify our later reductions.

\begin{lemma}\label{lem:fc_to_lda}
Consider a matrix $M \in \R^{n \times n}$. There exists a randomized Turing reduction from the problem of checking whether $M$ has full rank to the Linear Decision $(1/n^{O(1)})$-Approximation problem. The reduction runs in time $\tO(\nnz(M))$, produces $\plog(n)$ instances of the form $(M' \in \mathbb{R}^{n \times n} ,\mathbf{1}^n)$ where in the YES case $M'$ is a full rank matrix, and works w.h.p.
\end{lemma}

\begin{proof}
Let $M$ be an $n \times n$ matrix, and suppose that we need to check if it has full rank or not. We will reduce this problem to a linear system, by sampling a random vector $b \sim \N(0,1/n)^n$ and checking satisfiability of the linear system $(M,b)$. We will now show the following properties:
\begin{enumerate}
    \item When $M$ has full rank, then the linear system is satisfiable i.e. there exists an $x$ such that $Mx = b$.\label{prop:1}
    \item When $M$ does not have full rank, then for all $x \in \R^n$, $\norm{Mx - b}_2 \geq \norm{b}/n^{12}$ w.p. $1-O(1/\sqrt{n})$.\label{prop:2}
\end{enumerate}

Property~\ref{prop:1} above holds for the linear system since $M$ has full column span if it has full rank, which means that $Mx = b$ is satisfiable for all $b \in \R^n$. We will now prove Property~\ref{prop:2}, by showing that w.p. $1 - O(1/\sqrt{n})$, $b$ has a large component in $\colspace(M)^{\perp}$ i.e. the subspace orthogonal to $\colspace(M)$.

If $M$ does not have full rank, then $\rk(\colspace(M)) < n$. Let $V = \colspace(M)$ and $W = V^{\perp}$; since $V$ has dimension $< n$, we have that $\dim(W) \geq 1$. So there exists a unit vector $w \in \R^n$, with $\norm{w} = 1$, such that $\iprod{w, a} = 0, \forall a \in \colspace(M)$. From Lemma~\ref{lem:solns}, we have that any solution $x \in \R^n$ will have error, $\norm{Mx - b} \geq \norm{P_W(b)} \leq \abs{\iprod{w,b}}$ as $w \in W$ and $w$ is a unit vector. As $b \sim \N(0,1/n)^n$ we have that $\iprod{w,b}$ is distributed as the Gaussian $\N(0,\norm{w}_2^2/n) = \N(0,1/n)$ by Lemma~\ref{lem:gaussian}. So we should expect $|\iprod{w,b}| \approx 1/\sqrt{n} \approx \norm{b}/\sqrt{n}$. Indeed we have that, 
\begin{align*}
    \Pr\left[\abs{\iprod{w , b}} \leq \frac{\norm{b}}{n^2} \right] &\leq 
    \Pr\left[\abs{\iprod{w , b}} \leq \frac{\norm{b}}{n^2} \mathrel{\Big|} \norm{b} \leq n \right] + \Pr[\norm{b} \geq n] \\
    &\leq \Pr\left[\abs{\iprod{w , b}} \leq \frac{1}{n} \right] + \Pr[\norm{b}^2 \geq n^2] \\
    &\leq \frac{\sqrt{n}}{n} + \Pr[\norm{b}^2 \geq n^2] \hspace{10pt}\text{ Using property~\ref{prop:gp1} of Lemma~\ref{lem:gaussian}}\\
    &\leq \frac{1}{\sqrt{n}}+\frac{1}{n^2} \hspace{10pt}\text{ Using property~\ref{prop:gp3} of Lemma~\ref{lem:gaussian}}\\
    &\leq O\left(\frac{1}{\sqrt{n}}\right)
\end{align*}
which proves property~\ref{prop:2}. 

We have now reduced the Full column rank problem to the linear system $(M,b)$, such that in the YES case, we have that $Mx = b$ and in the NO case, we have a gap of $1/n^2$. We will now show that one can reduce to a linear system $(M', \mathbf{1}^n)$. To get the $i^{th}$ row of $M'$, scale the $i^{th}$ row of $M$ by $1/b_i$, that is, $M'_{ij} = M_{ij}/b_i$ (we can do this as $b_i \neq 0$ whp). It is easy to see that if there was a solution $x$ for which $Mx = b$, then we also have that $M'x = \mathbf{1}^n$. 

We will now show that in the NO case, if for all $x \in \R^n$, we had that $\norm{Mx - b} \geq \norm{b}/n^2$, then for all $x \in \R^n$ it holds that $\norm{M'x - \mathbf{1}^n} \geq \norm{\mathbf{1}^n}/n^{12}$. First note that, by the properties of Gaussians in Lemma~\ref{lem:gaussian}, w.p. $\geq 1-O(1/\sqrt{n})$, every $b_i$ satisfies $1/n^2 \leq \abs{b_i} \leq n^2$. Note the if one of these conditions is not satisfied we can just consider our reduction to have returned the wrong answer. This adds $O(1/\sqrt{n})$ to our error probability which now in total is $O(1/\sqrt{n})+O(1/\sqrt{n})=O(1/\sqrt{n})$. Assuming every $b_i$ satisfies $1/n^2 \leq \abs{b_i} \leq n^2$ we have the following sequence of inequalities:

\begin{align*}
\norm{M'x - \mathbf{1}^n}_2^2 &= \sum_i ((M'x)_i - 1)^2 \\ 
&= \sum_i \frac{1}{b_i^2}((Mx)_i - b_i)^2 \\
&\geq \frac{1}{n^4}\norm{Mx - b}_2^2, ~~~~~~~\text{since }\forall i, |b_i| \leq n^2\\
&\geq \frac{\norm{b}^2}{n^8} \\
&\geq \frac{\norm{\mathbf{1}^n}^2}{n^{12}}, ~~~~~~~\text{since }\forall i, |b_i| \geq 1/n^2,
\end{align*}
which implies that for all $x \in \R^n$, $\norm{M'x - \mathbf{1}^n} \geq \norm{\mathbf{1}^n}/n^{12}$. 

This randomized algorithm currently works with probability $1 - O(1/\sqrt{n})$; one can amplify the success probability to $1 - 1/n^{\log n}$ by $\plog(n)$ calls to the $(1/n^{12}) = (1/n^{O(1)})$-approximate linear decision problem.
\end{proof}

We will now show that one can amplify the error in the NO case from $1/n^{O(1)}$ to $1 - 1/n^c$ for all constants $c$.

The following lemma works for all $\eps(n), \delta(n)$, but we only state it for $\eps(n) = 1/n^{O(1)}, \delta(n) = 1/n^c$ to maintain consistency with the later sections in which we will work over WordRAM.
\begin{lemma}\label{lem:amp}
For all constants $c$ and $\eps(n) = 1/n^{O(1)}, \delta(n) = 1/n^c$, there exists a deterministic many-one reduction from the $\eps(n)$-Approximate linear search problem on the linear system $(A \in \R^{n \times n}, \mathbf{1}^n)$ to the $(1 - \delta(n))$-Approximate linear search problem on the linear system $(A' \in \R^{n \times n}, \mathbf{1}^n)$, with $\nnz(A') = O(\nnz(A))$. The reduction runs in time $\tO(\nnz(A))$. Additionally if $A$ is full rank then the matrix $A'$ produced is also full rank. 

As this is a deterministic many-one reduction we also get a gap-amplifying reduction for the $\eps(n)$-Approximate linear decision problem with the same parameters.
\end{lemma}

\begin{proof}
We create the new linear system $(A' \in \R^{n \times n}, \mathbf{1}^{n})$, by combining equations from the linear system $(A,b)$ in a cyclic manner. Let the vectors $a_i$'s be the rows of $A$. Then, create the $i^{th}$-equation in the new linear system as follows:

\[ (t+1) \iprod{a_i,x} - t\iprod{a_{i+1},x} = 1,\]
where the addition wrt $i$ is modulo $n$.

It is easy to see that that any $x$ which satisfies the original linear system $(A, \mathbf{1}^{n})$ also satisfies the new linear system $(A', \mathbf{1}^{n})$. We have that $A' = M \cdot A$, where $M \in \R^{n \times n}$, and $$M = \begin{bmatrix}
(t+1) & -t & 0 &\ldots & 0 & 0\\
0 & (t+1) & -t &\ldots & 0 & 0\\
0 & 0 & (t+1) &\ldots & 0 & 0\\
& & & \vdots & \\
0 & 0 & 0 &\ldots & (t+1) & -t\\
-t & 0 & 0 &\ldots & 0 & (t+1)\\
\end{bmatrix}$$

i.e. $M_{ii} = t+1$, $M_{i,i+1} = -t$ and $M_{ij} = 0$ otherwise. It is easy to check that $M$ has full rank. Hence if $A$ had full rank then so does $A' = M \cdot A$. The operation of multiplying $M$ to a vector amplifies the components orthogonal to $1^n$ and does not modify the components along $1^n$. As our target vector is $1^n$ it remains constant under the application of $M$. We will now argue that pre-multiplying the linear system by $M$ amplifies the error.

Let $x \in \R^n$, with corresponding error vector: $z_x = (z_1, z_2, \ldots, z_n)$ i.e. $Ax-\mathbf{1}=z_x$ and $\iprod{a_i,x} - 1 = z_i, \forall i \in [n]$. Let $z_x = y_x-\gamma_x\mathbf{1}^n$ where $y_x \cdot \mathbf{1}^m = 0$. Let us calculate the error of $x$ in the new linear system:
\begin{align}
\begin{split}\label{eq:t}
    \norm{A'x-\mathbf{1}^{n}}^2 &= \sum_{i \in [n]} \left((t+1) \iprod{a_i,x} - t\iprod{a_{i+1},x} - 1\right)^2 \hspace{20pt} \text{[Addition in $i$ is cyclic.]}\\
    &= \sum_{i \in [n]} \left((t+1)z_i - tz_{i+1}\right)^2 \hspace{15pt} \text{ [since  $\iprod{a_i,x} - 1 = z_i$]}\\ 
    &= \sum_{i \in [n]} (t^2 + 2t + 1)z_i^2 + t^2 z_{i+1}^2 - 2t(t+1)z_iz_{i+1} \\
    &= \sum_{i \in [n]} t(t+1)(z_i^2+z_{i+1}^2)+\frac{1}{2}(z_i^2+z_{i+1}^2) + \left(t+\frac{1}{2}\right)(z_i^2-z_{i+1}^2) - 2t(t+1)z_iz_{i+1}\\
    &= \sum_{i \in [n]} t(t+1) (z_i-z_{i+1})^2+\sum_{i \in [n]}\frac{1}{2}(z_i^2+z_{i+1}^2) + \sum_{i \in [n]}\left(t+\frac{1}{2}\right)(z_i^2-z_{i+1}^2)\\
    &= \sum_{i \in [n]} t(t+1) (z_i-z_{i+1})^2+\sum_{i \in [n]}\frac{1}{2}(z_i^2+z_{i+1}^2) \hspace{15pt} \text{[As the addition is cyclic]}\\
    &= \sum_{i \in [n]}z_i^2+\sum_{i \in [n]} t(t+1) (z_i-z_{i+1})^2 \hspace{15pt} \text{[As the addition is cyclic]}\\
    &= \norm{z_x}^2 + t(t+1)\sum_{i \in [n]} (z_i-z_{i+1})^2\\
    &= \norm{z_x}^2 + t(t+1)\sum_{i \in [n]} (y_i-y_{i+1})^2 \hspace{25pt} \text{ [As } z_i = y_i-\gamma_x \text{ and } z_{i+1} = y_{i+1}-\gamma_x\text{]} \\
    &\geq \norm{z_x}^2+t(t+1)\norm{y_x}^2/n^2 \hspace{50pt} \text{[ By Lemma~\ref{lem:expander} and as $y \cdot \mathbf{1}^{n} = 0$]}\\
\end{split}
\end{align}

If the input instance $(A, \mathbf{1}^{n})$ is a satisfiable instance then there exists an $x$ such that $z_x = y_x = \mathbf{0}^n$ and also $\gamma_x = 0$ which implies that $$\norm{A'x-\mathbf{1}^{n}}^2 = \norm{z_x}^2 + t(t+1)\sum_{i \in [n]} (y_i-y_{i+1})^2 = 0$$
hence the resulting instance when starting from a satisfiable instance is a satisfiable instance. 

On the other hand if we started from an unsatisfiable instance of the we have that $\norm{Ax - \mathbf{1}^n} \geq \eps(n)\norm{\mathbf{1}^n}$ where $\eps(n) = 1/n^{O(1)}$.

Set $t = n/(\delta(n)\eps(n))$. Now suppose we have access to an $x'$ such that $\norm{A'x'-\mathbf{1}^n} \leq (1-\delta(n)) \norm{1^{n}}$ then we will prove that we can find $x$ such that $\norm{Ax-\mathbf{1}^n} \leq \eps(n) \norm{1^{n}}$ in $\tO(\nnz(A))$ time. This will imply the lemma statement.

Let $Ax'-\mathbf{1}^n = z_{x'} = y_{x'}-\gamma_{x'}\cdot\mathbf{1}^n$ where $\iprod{z_{x'}, \mathbf{1}^n} = 0$ then by Equation~\ref{eq:t} we know that $\norm{A'x'-\mathbf{1}^n}^2 \geq t(t+1)\norm{y_{x'}}^2/n^2$. As $\norm{A'x'-\mathbf{1}^n} \leq (1-\delta(n)) \norm{1^{n}}$ we get that $t^2\norm{y_{x'}}^2/n^2 \leq \norm{1^{n}}^2$ which gives $$\norm{y_{x'}} \leq n\norm{1^{n}}/t = \delta(n)\eps(n)\norm{1^{n}}$$ 

We also know that $\norm{A'x'-\mathbf{1}^n} \geq \norm{z_{x'}} \geq \abs{\gamma_{x'}}\norm{1^{n}}$. This combined with $\norm{A'x'-\mathbf{1}^n} \leq (1-\delta(n)) \norm{1^{n}}$ gives us that $$\abs{\gamma_{x'}} \leq 1-\delta(n)$$

Let $x = x'/(1-\gamma_{x'})$, which can be easily calculated in $\tO(\nnz(A))$ time. For $x$, the error in the original linear system $(A, \mathbf{1}^n)$ is 
\begin{align*}
    \norm{Ax-\mathbf{1}^n} &= \normm{\frac{Ax'}{1-\gamma_{x'}}-\mathbf{1}^n} \\
    &= \normm{\frac{y_{x'}+\mathbf{1}^n(1-\gamma_{x'})}{1-\gamma_{x'}}-\mathbf{1}^n} \hspace{15pt} \text{[As $Ax'-\mathbf{1}^n = y_{x'}-\gamma_{x'}\cdot\mathbf{1}^n$]}\\
    &= \normm{\frac{y_{x'}}{1-\gamma_{x'}}} \hspace{15pt}\\
    &\leq \normm{\frac{y_{x'}}{\delta(n)}} \hspace{15pt} \text{As $\abs{\gamma_{x'}} \leq 1-\delta(n)$}\\
    &\leq \eps(n)\norm{\mathbf{1}^n} \hspace{15pt} \text{As $\norm{y_{x'}} \leq \delta(n)\eps(n)\norm{1^{n}}$}\\
\end{align*}
\end{proof}

We can now combine all the lemmas above to get the proof of Theorem~\ref{thm:red-main-formal}.

\begin{proof}[Proof of Theorem~\ref{thm:red-main-formal}]
We have proved the following sequence of reductions which preserve sparsity:

\begin{enumerate}
    \item Lemma~\ref{lem:rank}: Linear decision problem $\leq_T$ Full rank problem
    \item Lemma~\ref{lem:fc_to_lda}: Full rank problem $\leq_T$ $1/n^{O(1)}$-Approximate linear decision problem
    \item Lemma~\ref{lem:amp}: $1/n^{O(1)}$-Approximate decision problem $\leq_T$ $(1 - 1/n^c)$-Approximate linear decision problem.
\end{enumerate}

Note that in the YES case of Full-Rank problem we have a square full rank matrix, and this is propagated through Lemma~\ref{lem:fc_to_lda} and Lemma~\ref{lem:amp} hence the final instance we produce has a full rank matrix in the YES case. Since each of these reductions work whp, one can compose them to get the theorem statement.
\end{proof}

\section{Hardness of finding \texorpdfstring{$L_2$}{L2}-Approximate solutions on Real RAM}\label{sec:l2}
In this section, we elaborate on the implications of our main reduction from the previous section (Theorem~\ref{thm:red-main-formal}). Below is the map of reductions we showed in the previous sections. We will introduce conjectures for the Rank-finding problem in this section and given the reductions, the conjectures will imply conditional hardness of the approximate linear search problem. 

All the results in this section are for the RealRAM but can be obtained over the WordRAM too. In Section~\ref{sec:l2_w}, we prove the conditional hardness of the approximate linear search problem over general matrices in the WordRAM model.

\begin{figure}[H]
\begin{tikzcd}[column sep=15ex, row sep=10ex]
\text{Rank-Finding} 
\arrow[d, leftrightarrow, "\text{Lemma}~\ref{lem:eq}"] 
\arrow[r, "\text{Lemma}~\ref{lem:rank}+\ref{lem:fc_to_lda}"] & 
\begin{tabular}{c}
$1/n^{O(1)}$-Approximate \\
Linear Decision
\end{tabular}
\arrow[d, "\text{Decision to Search (Lemma~\ref{lem:dec_to_search-r})}"] 
\arrow[r, "\text{Lemma}~\ref{lem:amp}"] &
\begin{tabular}{c}
$(1-1/n^{\Omega(1)})$-Approximate \\
Linear Decision
\end{tabular}
\arrow[d, "\text{Decision to Search}"] 
\\
\text{Linear Decision} &
\begin{tabular}{c}
$1/n^{O(1)}$-Approximate \\
Linear Search
\end{tabular}
\arrow[r, "\text{Lemma}~\ref{lem:amp}"] &
\begin{tabular}{c}
$(1-1/n^{\Omega(1)})$-Approximate \\
Linear Search
\end{tabular}
\end{tikzcd}
\caption{Reductions on RealRAM preserving sparsity (up to $\plog(n)$ factors) and dimension (up to constant factors).}
\end{figure}
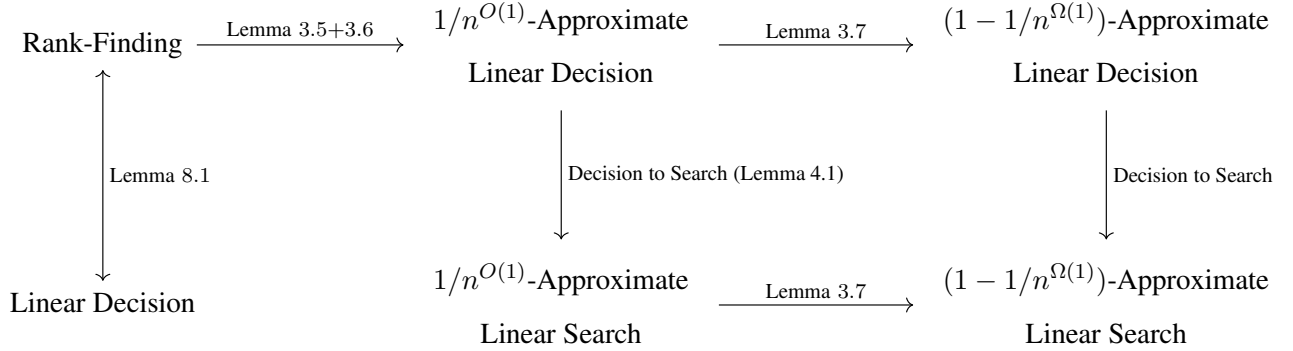

Note that all the lemmas pointed out here (except for the decision to search reduction) were shown in the previous section. The decision to search reduction is straightforward to carry out in the RealRAM and we formally prove it in Lemma~\ref{lem:dec_to_search-r}.

In the sections below we discuss the hardness of the approximate linear search problem over general matrices, and then the case of restricted classes of linear systems - sparse linear systems and linear systems given by positive-semidefinite matrices. Our results give
tight conditional hardness for all the problems considered.

\subsection{General linear systems}
We base the hardness result for the approximate linear search problem under the following conjecture for the rank-finding problem:

\begin{reminder}{Conjecture~\ref{conj:rk-dense}}
Finding the rank of a matrix $A \in \R^{m \times n}$ with $m = O(n)$ in the RealRAM model of computation is $\tOm(n^{\omega})$-hard. 
\end{reminder}

We now prove the decision to search reduction.

\begin{lemma}\label{lem:dec_to_search-r}
If there exists a $O(t(n))$ time algorithm to solve $\eps$-Approximate linear search problem for a linear system with sparsity $s$ then there exists a $O(t(n)+sn)$ time algorithm to solve $\eps$-Approximate linear decision problem for linear system with sparsity $s$.
\end{lemma}
\begin{proof}
We will follow the standard decision to search reductions which proceed by solving and then confirming. Suppose we are given a linear system $(A, b)$ with $\nnz(A) = s$ for which we want to solve $\eps$-Approximate linear decision problem. 

Suppose it is a YES instance i.e. there exists an exact solution, then by the assumed algorithm for $\eps$-Approximate linear search problem we can find an $x'$ such that $\norm{Ax'-b}_2 \leq \eps\norm{b}_2$. This can be done in $O(t(n))$ time. 

Suppose instead we were in the NO case i.e. for all $x'$, $\norm{Ax'-b}_2 > \eps\norm{b}_2$ i.e. there exists no $x'$ such that $\norm{Ax'-b}_2 \leq \eps\norm{b}_2$.

Hence checking whether $\norm{Ax'-b}_2 \leq \eps\norm{b}_2$ is true of not for the $x'$ returned by the assumed algorithm for $\eps$-Approximate linear search problem will let us solve the $\eps$-Approximate linear decision problem. We can check if $\norm{Ax'-b}_2 \leq \eps\norm{b}_2$ in time $O(sn)$. Hence the total running time is $O(t(n)+sn)$.

\end{proof}

Combining the main reduction (Theorem~\ref{thm:red-main-formal}) with a decision to search reduction (Lemma~\ref{lem:dec_to_search-r}) we get optimal conditional hardness for the $(1 - 1/n^{c})$-approximate linear search problem under Conjecture~\ref{conj:rk-dense}:

\begin{corollary}[Corollary~\ref{cor:dense} restated]\label{cor:dense-restated}
Under Conjecture~\ref{conj:rk-dense}, for all constants $c > 0$, the  $(1-1/n^c)$-approximate linear search problem is $\tOm(n^\omega)$ hard in the RealRAM model of computation. Moreover, this remains true even when the matrix $A$ in the given linear system $(A,b)$ is square and has full rank.
\end{corollary}

This conditional hardness result is tight as the best algorithms for (exactly) solving general linear systems run in time $\tO(n^\omega)$~\cite{bunch}.

The next corollary is a search to search reduction between the approximate linear search problem with small gap to one with a larger gap. This is a direct consequence of Lemma~\ref{lem:amp}. Recall that in the proof of Lemma~\ref{thm:red-main-formal} we used Lemma~\ref{lem:amp} to amplify the gap of the approximate linear \emph{decision} problem.  Lemma~\ref{lem:amp} is in fact more general and can amplify the gap of the approximate linear \emph{search} problem too (as noted in the lemma statement):

\begin{corollary}[Corollary~\ref{cor:search-amp} restated]\label{cor:search-amp-restated}
If for any constant $c > 0$ and $a$ there exists an $\tO(n^{a})$-time algorithm for $(1-1/n^c)$-Approximate Linear Search problem then for all constants $d$ there exists a $\tO(n^{a})$-time algorithms for $1/n^d$-Approximate Linear Search problem.
\end{corollary}

Even though this is stated as a reduction, one could potentially use the above corollary to get better algorithms for the $1/n^{O(1)}$-approximate linear search problem.

\subsection{Sparse linear systems}\label{sec:sparse_l2}
In this section, we give analogous results for sparse linear systems. Here we use the fact that our main reduction (Theorem~\ref{thm:red-main-formal}) preserves the sparsity of the original matrix. Hence if we assume the hardness of the rank-finding problem over sparse matrices, we get conditional hardness for approximately solving sparse linear systems.

\begin{reminder}{Conjecture~\ref{conj:rk-sparse}}
Finding the rank of a matrix $A \in \R^{m \times n}$, where $m = O(n)$ and $\nnz(A) = \tO(n)$, in the RealRAM model of computation, is $\tOm(n^2)$-hard.
\end{reminder}

Combining the main reduction (Theorem~\ref{thm:red-main-formal}) with a decision to search reduction we get optimal conditional hardness for the $(1 - 1/n^{c})$-approximate linear search problem on sparse matrices, under Conjecture~\ref{conj:rk-sparse}:

\begin{corollary}\label{cor:sparse}
Under Conjecture~\ref{conj:rk-sparse}, for all constants $c > 0$, the  $(1-1/n^c)$-approximate linear search problem $(A, b)$ with $\nnz(A) = \tO(n)$ is $\tOm(n^2)$ hard, in the RealRAM model of computation. Moreover, this remains true even when the matrix $A$ is square and has full rank.
\end{corollary}

Note that here we crucially used the fact that the main reduction preserves the sparsity of the original matrix. This conditional hardness result is tight as the best algorithms for (exactly) solving linear equations run in time $\tO(\nnz(A) \cdot n)$~\cite{hestenes} which is equal to $\tO(n^2)$ for sparse matrices.

Now we state the search to search reduction for the approximate linear search problem which follows from Lemma~\ref{lem:amp}. This reduction amplifies a small gap to a large gap. 

\begin{corollary}\label{cor:search-amp-sparse}
If for any constant $c > 0$ and $a$ there exists an $\tO(n^{a})$-time algorithm for the $(1-1/n^c)$-Approximate Linear Search problem over $(A,b)$ then for all constants $d$ there exists an $\tO(n^{a})$-time algorithm for the $1/n^d$-Approximate Linear Search problem over $(A',b)$ where $\nnz(A') = \nnz(A)$.
\end{corollary}

\subsection{Positive semidefinite linear systems}\label{sec:psd}

In this section, we show hardness for dense linear systems over PSD matrices. To do so, we need some additional ideas beyond our main reduction and also a conjecture for solving linear systems on matrices with intermediate sparsities. This conjecture also allows us to show optimal conditional hardness of approximately solving linear systems $(A,b)$ for any sparsity.

\begin{conjecture}[Rank-finding conjecture for all sparsities]\label{conj:rk-mid}
Finding the rank of a matrix $A \in \R^{m \times n}$, where $m = O(n)$, is $\min(\tOm(\nnz(A) \cdot n), \tOm(n^\omega))$ hard in the RealRAM model of computation.
\end{conjecture}

The current best known algorithms for the Rank-Finding problem $(A \in \mathbb{R}^{O(n) \times O(n)})$ runs in time $\min(\nnz(A)n, n^{\omega})$. The above conjecture assumes that this is optimal. We can prove the following theorem directly by combining Conjecture~\ref{conj:rk-mid} and Lemma~\ref{lem:red-main_w}.

\begin{corollary}\label{cor:mid}
Under Conjecture~\ref{conj:rk-mid}, for all constants $c > 0$, $(1-1/n^c)$-approximate linear search problem $(A, b)$, is $\min(\tOm(\nnz(A) \cdot n), \tOm(n^\omega))$ hard in the RealRAM model of computation. Moreover, this remains true even when the matrix $A$ is square and has full rank.
\end{corollary}

The next lemma reduces the approximate linear search problem for intermediate sparsities to approximate linear search problem on dense PSD matrices. It's proof exploits the fact that the algorithm of Yuster and Zwick~\cite{YusterZ05} does matrix multiplication of two matrices $A, A'$ in $O(\min(n^{\omega}, z^{.7}n^{1.2}+n^2))$ time where $z \geq \nnz(A)$ and $z > \nnz(A')$. This is faster than the best algorithm for solving linear systems $(A, b)$ which runs in $O(\min(n^{\omega}, zn))$ where $z = \nnz(A)$ for certain sparsities. Specifically we will use the following result from Yuster and Zwick~\cite{YusterZ05}:

\begin{theorem}[Yuster and Zwick~\cite{YusterZ05}, See Theorem 3.1 and discussion]\label{thm:yz}
    Assuming $\omega > 2$ there exists constants $.1 \geq \gamma, \gamma' > 0$ such that for two matrices $A, B \in \mathbb{R}^{O(n) \times O(n)}$ which satisfy $\nnz(A), \nnz(B) \leq n^{\frac{\omega+1}{2}-\gamma}$ can be multiplied in time $O(n^{\omega-\gamma'})$.
\end{theorem}

This allows us to do the following reduction:

\begin{lemma}\label{lem:psd_dense}
Assuming $\omega > 2$, there exists constants $.1 \geq \gamma, \gamma' > 0$ and a reduction running in time $O(\max(n^{\omega-\gamma'}, n^{\omega-\gamma}))$ from $(1-\delta(n))$-approximate linear search problem $(V \in \mathbb{R}^{n \times n}, b)$ with $\nnz(V) \leq n^{\frac{\omega+1}{2}-\gamma}$ to $(1-\delta(n))$-approximate linear search problem $(V', b)$ such that the matrix $V'$ is PSD.
\end{lemma}

\begin{proof}
The constants $\gamma, \gamma'$ are from Theorem~\ref{thm:yz}. By Theorem~\ref{thm:yz} we can compute $V '= VV^T$ in time $O(n^{\omega-\gamma'})$. 

We now prove that $\min_x \norm{Vx-b} = \min_x \norm{V'x-b} = \min_x \norm{VV^Tx-b}$. It is clear that $\min_x \norm{Vx-b} \leq \min_x \norm{VV^Tx-b}$ so we are only left with proving $\min_x \norm{VV^Tx-b} \leq \min_x \norm{Vx-b}$. Let $x^* = \arg\min_x \norm{Vx-b}$. Let $W$ be the subspace formed by the rows of $V$. We can split $x^*$ as $x^* = y + z$ where $y = P_W(x)$ and for all $w \in W$, $\iprod{z, w} = 0$. As $y \in W$ we can write it as $y = V^Tu$ for some $u$. Now we note that
\begin{align*}
    Vx^* &= Vy + Vz\\
    &= Vy \hspace{40pt} \text{[$Vz = \mathbf{0}$, As $z$ is orthogonal to rows of $V$]}\\
    &= VV^Tu \hspace{40pt} \text{[$y = V^Tu$]}\\
\end{align*}

Hence we have that $\min_x \norm{V'x-b} = \min_x \norm{VV^Tx-b} \leq \norm{VV^Tu-b} = \norm{Vx^*-b} \leq \min_x \norm{Vx-b}$.

Now consider the linear system $(V' = VV^T, b)$, as $\min_x \norm{Vx-b} = \min_x \norm{VV^Tx-b}$ we have that satisfiable instances are mapped to satisfiable instances. Further, given an $x$ such that $\norm{V'x-b} \leq (1-\delta(n))\norm{b}$ we can easily find $y = V^Tx$ in time $O(\nnz(V)) = O(n^{\frac{\omega+1}{2}-\gamma})$ for which $\norm{Vy-b} =  \norm{VV^Tx-b} \leq (1-\delta(n))\norm{b}$. Creating the linear system $(VV^T, b)$ takes $O(n^{\omega-\gamma'})$ time by Theorem~\ref{thm:yz}. Hence the total time taken by the reduction is $O(\max(n^{\omega-\gamma'}, n^{\omega-\gamma}))$.
\end{proof}

We now compose the above reduction with Conjecture~\ref{conj:rk-mid} to get the following tight conditional hardness.

\begin{corollary}\label{cor:dense_psd}
Under Conjecture~\ref{conj:rk-mid}, for all constants $c > 0$ $(1-1/n^c)$-approximate linear search problem $(A, b)$ where $A$ is restricted to be a PSD matrix, is $\tOm(n^\omega)$ hard in the RealRAM model of computation. Moreover, this remains true even when the matrix $A$ in the given linear system $(A,b)$ has full rank.
\end{corollary}
\begin{proof}
We will use the constants $\gamma, \gamma'$ from Theorem~\ref{thm:yz}.
By Corollary~\ref{cor:mid} we have that $(1-1/n^c)$-approximate linear search problem on matrices of sparsity $z$ is $\min(\tOm(z \cdot n), \tOm(n^\omega))$ hard. For $z = n^{\frac{\omega+1}{2}-\gamma}$ we have that $z \cdot n \geq n^{\omega}$ hence we get $\tOm(n^\omega)$ hardness. By Lemma~\ref{lem:psd_dense} we can reduce $(1-1/n^c)$-approximate linear search problem on matrices of sparsity $z$ to the $(1-1/n^c)$-approximate linear search problem $(V', b)$ such that the matrix $V'$ is PSD in time $O(\max(n^{\omega-\gamma}, n^{\omega-\gamma'}))$. As $\gamma$ and $\gamma'$ are constants greater than 0, we have that $(1-1/n^c)$-approximate linear search problem on PSD matrices is $\tOm(n^\omega)$ hard.
\end{proof}

\subsection{Well-Conditioned Linear Systems}\label{sec:approx-rank}

In this section, we show conditional hardness of approximately solving well-conditioned linear systems. The condition number of a full-rank square matrix is the ratio of its maximum and minimum eigenvalues (see Definition~\ref{def:cond-num}). If the entries of a matrix are all $O(\log n)$-bits then the condition number of this matrix is at most exponential in $n$ (this is true even for rectangular full column-rank matrices). Therefore linear systems over matrices with polynomially-bounded condition number could be significantly easier to solve than general linear systems.

For the case of certain restricted classes of matrices such as directed Laplacians, the algorithm of Cohen et al~\cite{CohenKKPPRS18} for the $\eps(n)$-approximate linear search problem runs in time $\tilde{O}(\nnz(A)\log(\ka(A)/\eps(n)))$ which is a near-linear time algorithm for $\ka(A) = \poly(n)$. This is a significant improvement over algorithms for directed Laplacian systems with no bound on the condition number (which run in time $\tO(n^\omega)$). 

But for general systems no such improvement is known!
Conjugate gradient~\cite{hestenes} runs in time $\tO(\nnz(A))$, when  $\kappa(A) = \poly\log n$, whereas when $\kappa(A) = \poly(n)$ this algorithm gives \emph{no improvement} over the algorithm for general matrices. 

We show that if we assume that the rank-finding problem is hard over well-conditioned matrices ($\kappa(A) = \poly(n)$), then the approximate linear search problem is hard to solve over well-conditioned linear systems. The proof goes along the same lines as that for general matrices: we show in Lemmas~\ref{lem:fc_to_lda_c} and \ref{lem:amp-c} that our main reduction (Theorem~\ref{thm:red-main-formal}) in fact preserves the condition number of our original matrix. Then as a corollary we obtain the conditional hardness for approximately solving well-conditioned linear systems. Note that we show all the results here over the RealRAM but they can be easily modified to work on the WordRAM too.

\begin{figure}[H]
\begin{tikzcd}[column sep=15ex, row sep=10ex]
\begin{tabular}{c}
Well-conditioned \\
Rank-Finding
\end{tabular}
\arrow[r, "\text{Lemma}~\ref{lem:fc_to_lda_c}"] & 
\begin{tabular}{c}
Well-conditioned \\
$1/n^{O(1)}$-Approximate \\
Linear Decision
\end{tabular}
\arrow[d, "\text{Decision to Search}"]
\arrow[r, "\text{Lemma}~\ref{lem:amp-c}"] &
\begin{tabular}{c}
Well-conditioned \\
$(1-1/n^{\Omega(1)})$-Approximate \\
Linear Decision
\end{tabular}
\arrow[d, "\text{Decision to Search}"] 
\\
 &
\begin{tabular}{c}
Well-conditioned \\
$1/n^{O(1)}$-Approximate \\
Linear Search
\end{tabular}
\arrow[r, "\text{Lemma}~\ref{lem:amp-c}"] &
\begin{tabular}{c}
Well-conditioned \\
$(1-1/n^{\Omega(1)})$-Approximate \\
Linear Search
\end{tabular}
\end{tikzcd}
\caption{Reductions on the RealRAM preserving sparsity (up to $\plog(n)$ factors), dimension (up to constant factors) and condition-number (upto $\poly(n)$ factors). One difference from the results from the previous sections is that we no longer have the equivalence for the Rank-Finding Problem and the Linear Decision problem for Well-conditioned matrices.} 
\end{figure}
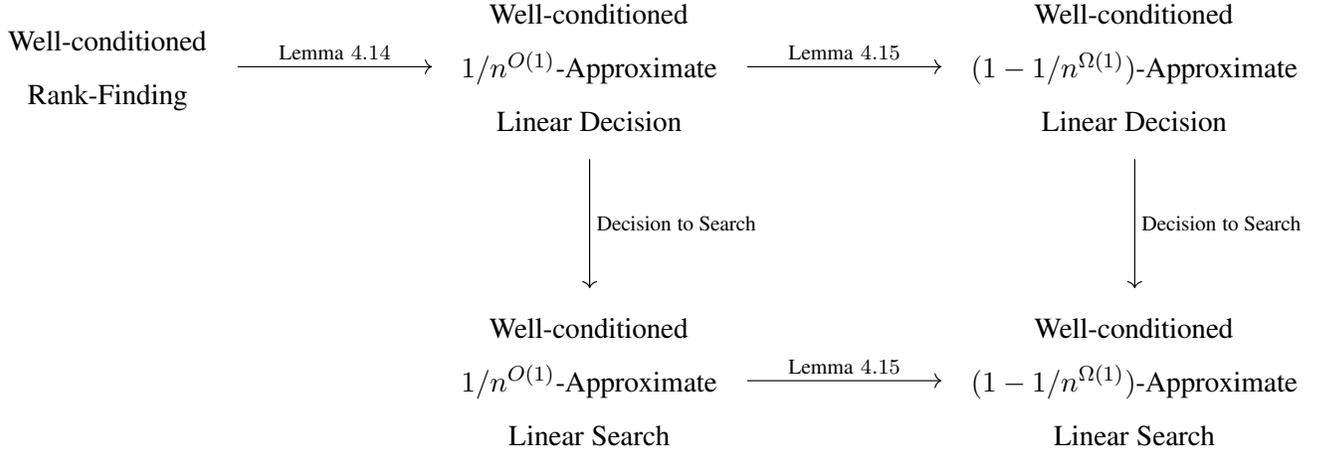

Before diving into the reduction, we will formally define all the problems used in the reduction-map above. As discussed in the definition of the condition-number, the condition-number is bounded only when the matrix has full column-rank, therefore in all the definitions below the matrices considered have dimension $m \times n$ with $n \leq m$.

\begin{definition}[Well-conditioned Rank-Finding Problem]\label{def:wcrf}
Given a matrix $A \in \R^{m \times n}$ where $n \leq m \leq O(n)$, with the promise that it falls into one of the following two sets of instances:
\begin{enumerate}
\item YES instances: $\rk(A) = n$ and $\kappa(A) \leq \poly(n)$.
\item NO instances: $\rk(A) < n$
\end{enumerate}
decide whether $A$ is a YES instance or a NO instance.
\end{definition}

\begin{definition}[Well-conditioned Full Column-Rank Problem]\label{def:wcfr}
Given a square matrix $A \in \R^{m \times n}$, where $n \leq m$, with the promise that it falls into one of the following two sets of instances:
\begin{enumerate}
\item YES instances: $\text{col-rank}(A) = n$ and $\kappa(A) \leq \poly(n)$.
\item NO instances: $\text{col-rank}(A) < n$
\end{enumerate}
decide whether $A$ is a YES instance or a NO instance.
\end{definition}

Note that the Well-conditioned Rank-finding problem easily reduces to the Well-conditioned Full Column-rank problem. This is because the YES instances of the former always have full column-rank by definition. In fact, we can also reduce the well-conditioned rank-finding problem to the well-conditioned \emph{full-rank} problem on \emph{square} matrices, by adding random columns, since this operation preserves the condition-number~\cite{edelman}. For simplicity of presentation we do not perform this operation and continue to work with full column-rank and (possibly) rectangular matrices throughout.

Next we formally introduce the search and decision problems on well-conditioned linear systems:

\begin{definition}[Well-conditioned $\eps(n)$-Approximate Linear Search Problem]\label{def:wc-e-als}
For a function $\eps: \mathbb{N} \rightarrow [0,1]$, given a satisfiable linear system $(A, b)$ with $A \in \mathbb{Z}^{m \times n}$ for $n \leq m$ and $\ka(A) = \poly(n)$ find an assignment $x$ such that $\norm{Ax-b} \leq \eps(n)\norm{b}$
\end{definition}

\begin{definition}[Well-Conditioned $\eps(n)$-Approximate Linear Decision problem]\label{def:approx_lin_c}
For a function $\eps: \mathbb{N} \rightarrow [0,1]$, given a linear system $(A \in \Z^{m \times n}, b \in \Z^n)$ for $n \leq m$, with the promise that it falls into one of the following two sets of instances:
\begin{enumerate}
\item YES instance: There exists an $x \in \Q^n$ such that $Ax = b$ and $A$ is well-conditioned.
\item NO instances: For all $x \in \R^n$, $\norm{Ax-b}_2 > \eps(n) \norm{b}_2$,
\end{enumerate}
decide whether $(A,b)$ is a YES instance or a NO instance.
\end{definition}

We will now show our main reduction from the Well-conditioned Rank-finding problem to the Well-conditioned $(1 - 1/n^{c})$-approximate linear decision problem. As noted above the Well-conditioned Rank-finding problem reduces to the Well-conditioned Full Column-Rank problem, so we will show a reduction from the latter to the approximate linear decision problem. To do so, we will show that the proofs in Section~\ref{sec:main-red} preserve the condition-number of the original matrix.

Let us start with showing that ``well-conditioned'' property is preserved in the reduction in Lemma~\ref{lem:fc_to_lda}.

\begin{lemma}\label{lem:fc_to_lda_c}
There exists a randomized Turing reduction from the Well-conditioned Full Column-Rank Problem on $M \in \Z^{m \times n}$ to the Well-conditioned $(1/n^{O(1)})$-Approximate Linear Decision problem. The reduction produces $\plog(n)$ instances of the form $(M',\mathbf{1}^n)$ where $M' \in \Z^{m \times n}$ and $\kappa(M') = \poly(n)$, runs in time $\tO(\nnz(M))$, and works w.h.p.
\end{lemma}

\begin{proof}[Proof Sketch]
The reduction and the proof is the same as the reduction in Lemma~\ref{lem:fc_to_lda}. The only extra thing we need to prove is that the resulting instances have $\poly(n)$ condition number. As in all the resulting instances $(M', b)$ we have that $M'$ is just a rescaling of rows of $M$ with the absolute value of the rescaling factor being between $1/n^2$ and $n^2$. This operation can only change the condition number by at most a multiplicative factor of $n^2 \cdot n^2 = n^4$. Hence $\ka(M') \leq \ka(M)n^4 = \poly(n)$.
\end{proof}

Next let us show that the ``well-conditioned'' property is preserved in the reduction in Lemma~\ref{lem:amp}.

\begin{lemma}\label{lem:amp-c}
For all constants $c,d > 0$, there exists a deterministic many-one reduction from the Well-conditioned $1/n^d$-Approximate linear search problem on the linear system $(A \in \Z^{m \times n}, \mathbf{1}^n)$ to the Well-conditioned $(1 - 1/n^c)$-Approximate linear search problem on the linear system $(A' \in \Z^{n \times n}, \mathbf{1}^n)$, with $\nnz(A') = O(\nnz(A))$ and $\kappa(A') = \poly(n)$.

As this is a deterministic many-one reduction we also get a gap-amplifying reduction for the $\eps(n)$-Approximate linear decision problem with the same parameters.
\end{lemma}
\begin{proof}[Proof Sketch]
    The only operation performed in Lemma~\ref{lem:amp} is that the input matrix $A$ is left multiplied with the ``$M$'' matrix in the proof of Lemma~\ref{lem:amp}. It is easy to verify that $\ka(M) = \poly(n^{d+c+1})$. So we get that $\kappa (A') = \ka(MA) \leq \ka(M)\ka(A) = \poly(n)$.
\end{proof}

Combining the two lemmas above we get the well-conditioned analogue of the main reduction from the Rank-finding problem to the $(1 - 1/n^{c})$-approximate linear decision problem. We can now apply a decision to search reduction to get conditional hardness for approximately solving linear systems over well-conditioned matrices:

\begin{corollary}\label{cor:dense_c}
For all constants $c$, assuming $\tOm(n^{\omega})$ hardness of the well-conditioned rank-finding problem we get that the well-conditioned $(1-1/n^c)$-approximate linear search problem is  $\tOm(n^{\omega})$-hard. 
\end{corollary}

We also state the following search to search reduction which follows directly from Lemma~\ref{lem:amp-c}: 

\begin{corollary}\label{cor:search-amp_c-restated}
For all constants $a,c,d > 0$, if there exists an $\tO(n^{a})$ time algorithm for well-conditioned $(1-1/n^c)$-approximate linear search then there exists an $\tO(n^{a})$-time algorithm for the well-conditioned $(1/n^d)$-approximate linear search problem.
\end{corollary}

\section{Hardness of Finding \texorpdfstring{$L_p$}{Lp}-approximate solutions}\label{sec:lp}
In this section, we will consider the following generalization of the $\eps(n)$-Approximate linear search problem to other norms:

\begin{definition}[$\eps(n)$-Approximate Linear Search over $L_p$-norm]\label{def:lp-e-als}
For a function $\eps: \mathbb{N} \rightarrow [0,1]$, given a satisfiable linear system $(A\in \R^{O(n) \times n}, b)$, find an $x \in \R^n$ such that $\norm{Ax-b}_p \leq \eps(n)\norm{b}_p$. 
\end{definition}

As in the case of $L_2$ norm, a 1-approximation is trivially achieved by $x = \mathbf{0}^m$ and we will prove that doing barely better than 1-approximation is conditionally hard. We first prove a search to search reduction between the approximate-linear-search problem in the $L_2$ and the $L_p$ norms (Lemma~\ref{lem:l_p}). Then combining with the conditional hardness for the $L_2$ search problem we get conditional hardness for the $L_p$ problem. 

\begin{corollary}\label{cor:dense_p}
Under conjecture~\ref{conj:rk-dense}, for every constant $c > 0$ and $p \geq 1$, the $(1-1/n^c)$-approximate linear search problem over the $L_p$-norm is $\tOm(n^{\omega})$ hard. Moreover, this remains true even when the matrix $A$ has full rank.
\end{corollary}

\begin{proof}
The proof follows directly by combining Corollary~\ref{cor:dense-restated} and Lemma~\ref{lem:l_p}.
\end{proof}

Let us now prove the main reduction of this section:

\begin{lemma}\label{lem:l_p}
Let $(A_{m \times n}, \mathbf{1^m})$ be a linear system with $m = O(n)$. Given an $x'$ such that $\norm{Ax' -\mathbf{1^m}}_p < (1-1/n^{c/3})\norm{\mathbf{1^m}}_p$ we can create an $x$ in time $O(\nnz(A))$ such that $\norm{Ax-\mathbf{1^m}}_2 < (1-1/n^c)\norm{\mathbf{1^m}}_2$.
\end{lemma}
\begin{proof}
Let $y = Ax' = \alpha\mathbf{1^m}+u$ where $u \cdot \mathbf{1^m} = 0$. Note that 
 \begin{align*}
     \norm{Ax'-\mathbf{1}}_p &= \norm{\alpha\mathbf{1}+z-\mathbf{1}}_p\\
     &= \norm{(\alpha-1)*\mathbf{1}+u}_p\\
     &\geq \norm{(\alpha-1)*\mathbf{1}}_p \hspace{30pt} \text{[By Claim~\ref{claim:p_all1} as $p \geq 1$ and $u \cdot \mathbf{1^m} = 0$]}\\
     &= \abs{1-\alpha}\norm{\mathbf{1}}_p.
 \end{align*}
Combining this with $\norm{Ax'-\mathbf{1}}_p < (1-1/n^{c/3})\norm{\mathbf{1}}_p$ gives us that $1/n^{c/3} \leq \alpha \leq 2$, which implies $\abs{\alpha} \leq 2$. We also have:

 \begin{align*}
     \norm{Ax'-\mathbf{1}}_p &= \norm{\alpha\mathbf{1}+z-\mathbf{1}}_p\\
     &= \norm{(\alpha-1)*\mathbf{1}+u}_p\\
     &\geq \norm{u}_p - \norm{(\alpha-1)*\mathbf{1}}_p \hspace{10pt} \text{[By triangle inequality for $L_p$ norm]}\\
     &\geq \norm{u}_2/n - \abs{1-\alpha}\norm{\mathbf{1}}_p \hspace{10pt} \text{[As $\norm{u}_p \geq \norm{u}_2/n$ for all $p \geq 1$]}
 \end{align*}
This implies that $\norm{u}_2 \leq n(\norm{Ax'-\mathbf{1}}_p+\abs{1-\alpha}\norm{\mathbf{1}}_p) \leq 2n\norm{\mathbf{1}^m}_p \leq n^2 \leq n^{c/6}$ as $m = O(n)$ and $c$ is a large enough constant.

Now using $\abs{\alpha} \leq 2$ and $\norm{u}_2 \leq n^{c/6}$ we can apply Claim~\ref{claim:lin-eq-2} (with $\eps(n) = 2/n^c, b = \mathbf{1}^m$) to find an $x$ in time $O(\nnz(A))$ such that $\norm{Ax - \mathbf{1}^m}_2 < (1-1/n^c)\norm{\mathbf{1}^m}_2$ 
\end{proof}

We now prove the outstanding claims that were used in the proof of the main lemma above.

\begin{claim}\label{claim:lin-eq-2}
Let $(A \in \mathbb{R}^{m \times n}, b \in \mathbb{R}^m, m = O(n))$ be a linear system. Let $x' \in \mathbb{R}^n$ and let $Ax' = \alpha \cdot b+u$ where $\iprod{u, b} = 0$ such that
    $$\norm{u} < \norm{b}\abs{\alpha}\sqrt{\frac{1-\eps(n)}{\eps(n)}}$$

Given $x'$ we can find an $x$ such that $\norm{Ax - b} < (1-\eps(n)/2)\norm{b}$ in time $O(\nnz(A))$.
\end{claim}

\begin{proof}
Let $x = (\eps(n)/\alpha)x'$, then
\begin{align*}
    \norm{Ax-b}^2 &= \norm{(\alpha (\eps(n)/\alpha)-1)\cdot b + (\eps(n)/\alpha)u}^2\\
    &= (\eps(n)-1)^2\norm{b}^2 + (\eps(n)/\alpha)^2\norm{u}^2 \hspace{15pt} \text{[As $\iprod{u, b} = 0$]}\\
    &< (1-\eps(n))\norm{b}^2 \hspace{30pt} \left[\text{As $\norm{u} < \norm{b}\abs{\alpha}\sqrt{\frac{1-\eps(n)}{\eps(n)}}$}\right]\\
\end{align*}
which implies that $\norm{Ax - b} < (1-\eps(n)/2)\norm{b}$.

As we can calculate $\alpha$ is time $O(\nnz(A))$, we can find $x$ in time $O(\nnz(A))$.
\end{proof}

\begin{claim}\label{claim:p_all1}
Let $u \in \mathbb{R}^n$ be a vector such that $\iprod{u, \mathbf{1}^n} = 0$. Then for all $p \geq 1$ we have that $\norm{\gamma \cdot \mathbf{1}^n + u}_p \geq \norm{\gamma \cdot \mathbf{1}^n}_p$.
\end{claim}
\begin{proof} 
Let $z = \norm{\gamma \cdot \mathbf{1}^n}_p$, consider the closed $p$-norm ball $(\sum_{i=1}^n \abs{x_i}^p)^{1/p} \leq z$. The point $\gamma \cdot \mathbf{1}^n$ lies on the boundary of this $p$-norm ball with the tangent hyperplane being $\iprod{v, \mathbf{1}^n} = 0$. We see that $u$ lies on the tangent plane and as the $p$-norm ball is convex we have that the point $\gamma \cdot \mathbf{1}^n + u$ lies outside the $p$-norm ball. Hence $\norm{\gamma \cdot \mathbf{1}^n + u}_p \geq \norm{\gamma \cdot \mathbf{1}^n}_p$.
\end{proof}

\section{Hardness of finding \texorpdfstring{$L_2$}{L2}-Approximate solutions on Word RAM}\label{sec:l2_w}

\begin{figure}[H]
\begin{tikzcd}[column sep=15ex, row sep=10ex]
\text{Rank-Finding} 
\arrow[d, leftrightarrow, "\text{Lemma}~\ref{lem:eq_w}"] 
\arrow[r, "\text{Lemma}~\ref{lem:rank_w}+\ref{lem:fc_to_lda_w}"] & 
\begin{tabular}{c}
$1/n^{O(1)}$-Approximate \\
Linear Decision
\end{tabular}
\arrow[d, "\substack{\text{Decision to Search}\\\text{Lemma}~\ref{lem:dec_to_search-w}}"]
\arrow[r, "\text{Lemma}~\ref{lem:amp-w}"] &
\begin{tabular}{c}
$(1-1/n^{\Omega(1)})$-Approximate \\
Linear Decision
\end{tabular}
\arrow[d, "\substack{\text{Decision to Search}\\\text{Lemma}~\ref{lem:dec_to_search-w}}"] 
\\
\text{Linear Decision} &
\begin{tabular}{c}
$1/n^{O(1)}$-Approximate \\
Linear Search
\end{tabular}
\arrow[r, "\text{Lemma}~\ref{lem:amp-w}"] &
\begin{tabular}{c}
$(1-1/n^{\Omega(1)})$-Approximate \\
Linear Search
\end{tabular}
\end{tikzcd}
\caption{Reductions on WordRAM preserving sparsity (up to $\plog(n)$ factors), dimension (up to constant factors) and bit complexity of entries (up to constant factors).}
\end{figure}
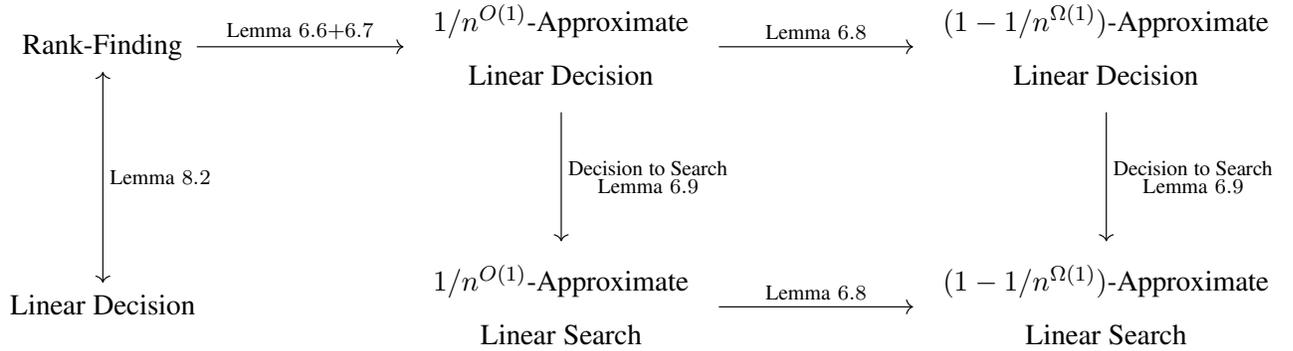

In Section~\ref{sec:l2} we gave conditional hardness for solving linear equations in RealRAM. In this section we will give hardness for WordRAM through modifications of the reduction for RealRAM.
We will assume that for a matrix/vector of dimension $m \times n$ in this section that all input entries are from $\mathbb{Z}$ and have $O(\log mn)$ bits. 

The reduction over the RealRAM (Theorem~\ref{thm:red-main-formal}) does not directly work for the WordRAM but the following modifications are sufficient for the reduction to go through:
\begin{enumerate}
    \item As we are in the WordRAM we need to argue that the final problem instance has bounded bit complexity if the starting problem has bounded bit complexity. To verify this we show that this is true for all the steps of the reduction.
    \item We sampled a random gaussian vector in the reduction (Lemma~\ref{lem:fc_to_lda}), this is not possible in the WordRAM. We will get around this issue (in Lemma~\ref{lem:fc_to_lda_w}) by sampling a random vector whose each entry is a random integer from a predefined range.
    \item The decision to search reduction for the RealRAM (Lemma~\ref{lem:dec_to_search-r}) was nearly trivial. This is not the case for the WordRAM as the solution can have large bit complexity and hence given a solution directly substituting to check if it is a good solution or not may require too much time (See Proof of Lemma~\ref{lem:dec_to_search-w} for more discussion). We give an alternative decision to search reduction in the WordRAM in Lemma~\ref{lem:dec_to_search-w}.
\end{enumerate}

We start by defining the Linear Decision Problem over Word-RAM:

\begin{definition}[Linear Decision Problem over Word-RAM]\label{def:ldp_w}
Given a linear system $(A, b)$ with $A \in \mathbb{Q}^{m \times n}$ with $\bc(A) = O(\log mn)$, distinguish between the following two sets of instances:
\begin{itemize}
    \item YES instances: There exists an $x$ such that $Ax = b$.
    \item NO instances: For all $x$, $Ax \neq b$.
\end{itemize}
\end{definition}

We will consider the following conjecture (analogous to Conjecture~\ref{conj:rk-dense}) for the Word RAM:

\begin{conjecture}[Rank-Finding Conjecture on WordRAM]\label{conj:rk-dense-w}
There exists no $\tilde{o}(n^{\omega})$-time randomized algorithm for finding the rank of a matrix $A \in \Z^{m \times n}$ with $m = O(n)$ and $\bc(A) = O(\log n)$, in the WordRAM model of computation.
\end{conjecture} 

The conjecture is tight as using an easy randomized reduction~\cite{Storjohann05} Rank over integers with $O(\log n)$ bit complexity can be reduced to rank over finite fields $\text{GF}(\poly(n))$ which gives a $\tO(n^\omega)$ time randomized algorithm on WordRAM.

We will prove the following main theorem:

\begin{corollary}[Hardness of solving linear equation on WordRAM]\label{cor:dense_w}
For all constant $c > 0$, under Conjecture~\ref{conj:rk-dense}, there does not exist a $\tilde{o}(n^{\omega})$ time randomized algorithm for the $(1-1/n^c)$-approximate linear search problem $(A, b)$ with $\bc(A) = O((c+1)\log n)$ in the WordRAM model of computation. Moreover, this remains true even when the matrix $A$ in the given linear system $(A,b)$ has full rank.
\end{corollary}

The conditional hardness in the above theorem is tight as there exists $O(n^\omega)$ time algorithms for exactly solving linear equations over $\Z$ in WordRAM~\cite{Storjohann05,PS12,BirmpilisLS19}.

We are also able to establish the following corollary of one of the intermediate steps in our reduction (Lemma~\ref{lem:amp-w}) which reduces approximately solving linear systems with low error to approximately solving linear systems with barely non-trivial error. \\
\begin{corollary}\label{cor:search-amp_w}
For all constant $c,d > 0$, if over WordRAM there exists as $\tO(n^{a})$ time algorithm for $(1-1/n^c)$-Approximate Linear Search then there exists a $\tO(n^{a})$-time algorithms for $(1/n^d)$-Approximate Linear Search problem.
\end{corollary}

Analogous to Theorem~\ref{thm:red-main-formal}, we will prove a reduction from the rank finding problem to approximate linear decision problem from which Corollary~\ref{cor:dense_w} will follow by a decision to search reduction (Lemma~\ref{lem:dec_to_search-w}). Formally,

\begin{lemma}[Reduction from exact to approximate]\label{lem:red-main_w}
For all constant $c$, there exists a randomized Turing reduction from the rank-finding problem on $A \in \Z^{m \times n}$ to the $(1 - 1/n^c)$-Approximate linear decision problem on $n' \times n'$-square matrices with bit complexity $\bc(A)=O((c+1)\log(n))$, with $n' = O(\max(m,n))$, where in the YES case we have the additional property that the matrices produced have full rank. The reduction runs in time $\tO(c \cdot \nnz(A))$, produces $\plog(mn)$ instances of the approximate linear decision problem and works w.h.p..
\end{lemma}

To prove the above Lemma, we will go through each of the steps in the proof of Theorem~\ref{thm:red-main} and see that all of them work for WordRAM with small modifications. We start with a modification of Lemma~\ref{lem:add_random}.
\begin{lemma}\label{lem:add_random_w} 
Let $M \in \Z^{m \times n}$ be a matrix of rank $\geq r$. Then there exists a sampling procedure to sample $z \in \Z^m$, with $\nnz(z) = \min(m, m\log^2(m)/(m-r +1))$ and $\bc(z) = O(\log \max(m, n))$, that runs in time $\tO(\nnz(z))$, such that the matrix $B = [M ~ z]$ has $\rk(B) \geq \min(r+1, m)$ with probability $\geq 1-1/(\max(m, n))^4$. 
\end{lemma}
\begin{proof}
We will follow the proof of Lemma~\ref{lem:add_random} except instead of Gaussians we will use a uniform distribution over integers $[1, N]$.

Let $s = \min(m, m\log^2(m)/(m-r+1))$. If $s = m$ let $S = [m]$, otherwise randomly choose $S \subseteq [m]$, by sampling $s$ coordinates from $[m]$ with replacement. If $i \notin S$, then set $z_{i} = 0$, else sample it from $[1, N]$ for $N = \max(m, n)^5$. It is easy to see that $\bc(z) = O(\log \max(m, n))$.

Now consider the matrix $B = [M ~ z]$. Suppose that $\rk(M) \geq \min(r+1, m)$, then we already have that $\rk(B) \geq \rk(M) \geq \min(r+1, m)$. So we need to only prove that when $\rk(M) = r < m$, then this procedure gives $\rk(B) \geq r+1$. This is equivalent to proving that $z \notin \colspace(M)$.

Let $W$ be the orthogonal subspace to $\colspace(M)$. $W$ has dimension $m-r > 0$ and for all vectors $w \in W, u \in \colspace(M)$, $\iprod{w,u} = 0$. We will prove that with large probability there exists a vector $v \in W$, such that $\iprod{v, z} > 0$, which implies that $z \notin \colspace(M)$. 

Since $\rk(W) = m - r$, by Lemma~\ref{lem:get_nz} there exists vector $v \in W$ which is non-zero on at least $m-r$ coordinates. Let $G$ denote the set of coordinates where $v$ is non-zero with $\abs{G} \geq m-r$.

Since $S$ (the set of non-zero coordinates of $z$) is a large enough random set of coordinates, one can prove that $S,G$ have a non-zero intersection. Formally,

\begin{align*}
    \Pr[S \cap G = \phi] &= \left(\Pr_{i \sim [m]}[i \not\in G]\right)^{\abs{S}}\\
    &\leq \left(r/m\right)^{m\log^2(m)/(m-r+1)}\\
    &\leq \frac{1}{m^{\log m}}.
\end{align*}

Let us now consider the inner product between $z,v$. We have that, $\iprod{z,v} = \sum_{i \in S \cap G} z_iv_i$. Assuming $S \cap G \neq \phi$ and $S \cap G = (i_1, i_2, \ldots, i_j)$ then for any given values of $z_{i_2}, z_{i_3}, \ldots, z_{i_j}$ there exists at most one value of $z_{i_1}$ among $[1, N]$ such that $\iprod{z,v} = 0$. Hence $$\Pr[\iprod{z,v} = 0 \mid S \cap G \neq \phi] \leq \frac{1}{N}$$ which in turn implies that:
\begin{align*}
    \Pr_{z \in [1, N]^m}[\iprod{z,v} = 0] &\leq \Pr[S \cap G = \phi] + \Pr[\iprod{z,v} = 0 \mid S \cap G \neq \phi]\\
    &\leq \frac{1}{m^{\log m}}+\frac{1}{N}\\
    &\leq \frac{1}{\max(m, n)^4} \hspace{70pt} \text{[As $N = 1/(\max(m, n))^5$]}
\end{align*}

$\iprod{z,v} \neq 0$ implies that $z \notin \colspace(M)$, and the matrix $B = [M ~ z]$ has rank equal to $\min(r+1, m)$. Hence the matrix $B = [M ~ z]$ has rank equal to $\min(r+1, m)$ with probability $\geq 1-1/(\max(m, n))^4$.
\end{proof}

Next we reduce the Rank-Finding Problem to the Full rank problem.

\begin{lemma}\label{lem:rank_w} There exists a randomized Turing  reduction which works w.h.p. from the Rank-Finding problem on $A \in \Z^{m \times n}$ with $m = O(n)$ and $\bc(A) = O(\log n)$, to the Full rank problem, that runs in time $\tO(\nnz(A))$ and produces $O(\log m)$ instances of the Full rank problem, such that all instances of the matrices produced have dimension $O(\max(m,n)) \times O(\max(m,n))$, sparsity $\tO(\nnz(A))$ and bit complexity $O(\log n)$.
\end{lemma}
 Given Lemma~\ref{lem:add_random_w}, the proof of the above Lemma~\ref{lem:rank_w} is the same as the proof of Lemma~\ref{lem:rank}.

Analogous to Lemma~\ref{lem:fc_to_lda}, we now reduce the Full-Rank Problem to the $(1/n^{O(1)})$-Linear Decision Approximation problem on WordRAM.

\begin{lemma}\label{lem:fc_to_lda_w}
Consider a matrix $M \in \Z^{n \times n}$. There exists a randomized Turing reduction from the problem of checking whether $M$ has full rank to the $(1/n^{12})$-Linear Decision Approximation problem. The reduction runs in time $\tO(\nnz(M))$, produces $\plog(n)$ instances of the form $(M',\mathbf{1}^n)$ where $M' \in \Z^{n \times n}, \bc(M') = O(\bc(M)+(\log n))$, in the YES case $M'$ is a full rank matrix, and works w.h.p..
\end{lemma}
\begin{proof}
Let $M$ be an $n \times n$ matrix, and suppose that we need to check if it has full rank or not. We will reduce this problem to a linear system, by sampling a random vector $b' \sim \N(0,1/n)^n$ rounding it to the nearest fraction of the form $a/N$ for $N = n^5$ to get a new vector $b$ i.e. $b_i = (\argmin_{a \in \Z}\abs{b'_i - (a/N)})/N$ and checking satisfiability of the linear system $(M, b)$. We will now show the following:
\begin{enumerate}
    \item When $M$ has full rank, then the linear system is satisfiable i.e. there exists an $x$ such that $Mx = b$.\label{prop:1_w}
    \item When $M$ does not have full rank, then for all $x \in \R^n$, $\norm{Mx - b}_2 \geq \norm{b}/n^2$ w.p. $1-O(1/\sqrt{n})$.\label{prop:2_w}
\end{enumerate}

Property~\ref{prop:1_w} above holds for the linear system since $M$ has full column span if it has full rank, which means that $Mx = b$ is satisfiable for all $b \in \R^n$. We will now prove the second point, by showing that w.p. $1 - O(1/\sqrt{n})$, $b$ has a large component in $\colspace(M)^{\perp}$

If $M$ does not have full rank, then $\rk(\colspace(M)) < n$. Let $V = \colspace(M)$ and $W = V^{\perp}$; since $V$ has dimension $< n$, we have that $\dim(W) \geq 1$. So there exists a unit vector $w \in \R^n$, with $\norm{w} = 1$, such that $\iprod{w, a} = 0, \forall a \in \colspace(M)$. From Lemma~\ref{lem:solns}, we have that any solution $x \in \R^n$ will have error, $\norm{Mx - b} \geq \norm{P_W(b)} \geq \abs{\iprod{w,b}}$ as $w \in W$ and $w$ is a unit vector. For large $N$, we expect that the distribution of $\iprod{w,b}$ is near to the Gaussian $\N(0,\norm{w}_2^2/n) = \N(0,1/n)$ by Lemma~\ref{lem:gaussian}. So we should expect $|\iprod{w,b}| \approx 1/\sqrt{n} \approx \norm{b}/\sqrt{n}$. Define $e = b-b'$ and note that for all $i$, $\abs{e_i} \leq 1/N$. Calculating exactly we get that $\iprod{w,b} = \iprod{w,b'}+\iprod{w,e'}$. This implies that $\iprod{w,b'} - \sqrt{n}/N \leq \iprod{w,b} \leq \iprod{w,b'} + \sqrt{n}/N$.
\begin{align*}
    \Pr\left[\abs{\iprod{w , b}} \leq \frac{\norm{b}}{n^2} \right] &\leq 
    \Pr\left[\abs{\iprod{w , b}} \leq \frac{\norm{b}}{n^2} \mathrel{\Big|} \norm{b} \leq n \right] + \Pr[\norm{b} \geq n] \\
    &\leq \Pr\left[\abs{\iprod{w , b}} \leq \frac{1}{n} \right] + \Pr[\norm{b}^2 \geq n^2] \\
    &\leq \Pr\left[\abs{\iprod{w , b'}+\iprod{w , e}} \leq \frac{1}{n} \right] + \Pr[\norm{b'}^2+\norm{e}^2 \geq n^2] \\
    &\leq \Pr\left[\abs{\iprod{w , b'}} \leq \frac{1}{n}+\frac{\sqrt{n}}{N} \right] + \Pr[\norm{b'}^2 \geq n^2-\frac{\sqrt{n}}{N}] \hspace{20pt} \text{[As $\abs{e_i} \leq 1/N$]} \\
    &\leq \Pr\left[\abs{\iprod{w , b}} \leq \frac{2}{n} \right] + \Pr[\norm{b}^2 \geq n^2/2] \\
    &\leq \frac{2\sqrt{n}}{n} + \Pr[\norm{b}^2 \geq n^2/2] \hspace{10pt}\text{ Using property~\ref{prop:gp1} of Lemma~\ref{lem:gaussian}}\\
    &\leq \frac{2}{\sqrt{n}}+\frac{2}{n^2} \hspace{10pt}\text{ Using property~\ref{prop:gp3} of Lemma~\ref{lem:gaussian}}\\
    &\leq O\left(\frac{1}{\sqrt{n}}\right)
\end{align*}
which proves property~\ref{prop:2_w}. 

The rest of the proof of this lemma proceeds in the same way as the proof of Lemma~\ref{lem:fc_to_lda}. 
\end{proof}

Finally we reduce the $(1/n^{O(1)})$-Linear Decision Approximation problem to the $(1-1/n^c)$-Linear Decision Approximation problem over WordRAM. This is straightforward to work out from the analogous Lemma~\ref{lem:amp} in Section~\ref{sec:l2}.

\begin{lemma}\label{lem:amp-w}
For all constants $c$ and $\eps(n) = 1/n^{O(1)}, \delta(n) = 1/n^c$, there exists a deterministic many-one reduction from the $\eps(n)$-Approximate linear search problem on the linear system $(A \in \R^{n \times n}, \mathbf{1}^n)$ to the $(1 - \delta(n))$-Approximate linear search problem on the linear system $(A' \in \R^{n \times n}, \mathbf{1}^n)$, with $\nnz(A') = O(\nnz(A))$. The reduction runs in time $\tO(\nnz(A))$. Additionally if $A$ is full rank then the matrix $A'$ produced is also full rank. 

As this is a deterministic many-one reduction we also get a gap-amplifying reduction for the $\eps(n)$-Approximate linear decision problem with the same parameters.
\end{lemma}
\begin{proof}[Proof Sketch]

 The only operation that is performed in the proof of Lemma~\ref{lem:amp} is that the matrix $A$ is multiplied with the ``$M$'' matrix. As the entries of $M$ are of the order $n/(\eps(n)\delta(n))$ our final bit complexity increases by an additive factor of $O(\log(n/(\eps(n)\delta(n))))$. 
 
\end{proof}

Now we are ready to prove Lemma~\ref{lem:red-main_w}.

\begin{proof}[Proof of Lemma~\ref{lem:red-main_w}]
    Similar to the proof of Theorem~\ref{thm:red-main}, The proof follows by composing Lemma~\ref{lem:rank_w}, Lemma~\ref{lem:fc_to_lda_w} and Lemma~\ref{lem:amp-w}  (for $\eps(n) = 1/n^{O(1)}$ and $\delta(n) = 1/n^c$). The bit complexity of the final matrix is $O(\log n)+O(\log(n/(\eps(n)\delta(n)))) = O((c+1)(\log n))$ and the running time is $\tO(c \cdot \nnz(A))$
\end{proof}

To prove Corollary~\ref{cor:dense_w} we need the following decision to search reduction:

\begin{lemma}\label{lem:dec_to_search-w}
Let $\eps(n) = 1/n^{O(1)}$, given an $A \in \Z^{m \times n}, x, b$ where $m = O(n), \bc(A), \bc(b) = \plog(n), \bc(x) = \tO(n) $ we can distinguish between:
\begin{enumerate}
    \item $\norm{Ax-b} \leq \eps(n) \norm{b}/2$.
    \item $\norm{Ax-b} \geq \eps(n) \norm{b}$.
\end{enumerate}
w.h.p. in time $\tO(\nnz(A)n)$.
\end{lemma}
\begin{proof}
    We start by noting that unlike the analogous decision-to-search reduction (Lemma~\ref{lem:dec_to_search-r}) in RealRAM we cannot proceed by directly evaluating $\norm{Ax-b}$ exactly as that would take time $\tO(\nnz(A)n^2)$ in WordRAM as we are only given that $\bc(x) = O(n)$. But we do not need to evaluate $\norm{Ax-b}$ exactly even approximately evaluating will be sufficient as long as the error is $< \eps(n) \norm{b}/2$.
    
    For any $w$ with entries of the form $a/N$ where $a \in \Z$ and $N = n^6$ we can evaluate $\iprod{Ax-b, w} = (wA)x-wb$ in time $O(n\nnz(A)$ by first computing $u \coloneqq wA$ and then evaluating $ux-vb$. We are able to get around $\tO(n^3)$ runtime as $x$ is being only multiplied to a $1 \times n$ vector rather than a $n \times n$ matrix.
    
    To approximately evaluate $\norm{Ax-b}$ we will make use of the standard fact (from Lemma~\ref{lem:gaussian}) that for $v \sim \N(0, 1/m)^m$ we have $\iprod{Ax-b, v} \sim \N(0, \norm{Ax-b}^2)$. We will sample $v_1, v_2, \ldots, v_t$ i.i.d. from $\N(0, 1/m)^m$, then by Lemma~1 of~\cite{laurent2000adaptive} we have that
    $$\Pr\left[\abss{\sum\limits_{i=1}^{t} \iprod{Ax-b, v_i}^2 - t\cdot\norm{Ax-b}^2} > \norm{Ax-b}^2(2\sqrt{tz}+2z)\right] \leq 2 \cdot exp(-z).$$
    Hence at $t = \log^6(n)$ and $z = \log^3(n)$ we have that w.h.p. ($\geq 1-1/n^{\log n}$) 
    $$\abss{\sum\limits_{i=1}^{t} \iprod{Ax-b, v_i}^2 - t\cdot\norm{Ax-b}^2} \leq \norm{Ax-b}^2t/(\log n).$$

    Let $w_i$ be the vector obtained by rounding each entry of $v_i$ to the nearest fraction of the form $a/N$ where $N = n^6$. Then $\norm{v_i-w_i} \leq 1/n^5$ and hence
    
    $$\abss{\sum\limits_{i=1}^{t} \iprod{Ax-b, w_i}^2 - t\cdot\norm{Ax-b}^2} \leq \norm{Ax-b}^2t/(\log n)+\norm{Ax-b}^2t/n \leq 2\norm{Ax-b}^2t/(\log n).$$
    
    Hence if $\norm{Ax-b} \leq \eps(n) \norm{b}/2$ we have that w.h.p.
    $$\sum\limits_{i=1}^{t} \iprod{Ax-b, w_i}^2  \leq t\cdot\eps^2(n) \norm{b}^2/4+2(\eps^2(n) \norm{b}^2/4)(t/(\log n)) \leq t\cdot\eps^2(n) \norm{b}^2/3.$$
    
    On the other hand, if $\norm{Ax-b} \geq \eps(n) \norm{b}$ we have that w.h.p.
    $$\sum\limits_{i=1}^{t} \iprod{Ax-b, w_i}^2  \geq t\cdot\eps^2(n) \norm{b}^2-2(\eps^2(n) \norm{b}^2)(t/(\log n)) \geq t\cdot\eps^2(n) \norm{b}^2/2.$$
    
    Hence we can distinguish between the two cases w.h.p in time $\tO(n\nnz(A))$ (as $t = \plog(n)$).

\end{proof}

Combining Lemma~\ref{lem:red-main_w} and Lemma~\ref{lem:dec_to_search-w} give us Corollary~\ref{cor:dense_w}:

\begin{proof}[Proof of Corollary~\ref{cor:dense_w}]
For all constants $c$, Composing Conjecture~\ref{conj:rk-dense-w} and Lemma~\ref{lem:red-main_w} gives us $\tOm(n^\omega)$ hardness of $(1 - 1/n^c)$-Approximate linear decision problem $(A \in \Z^{n \times n}, b)$ with bit complexity $O((c+1)\log(n))$ where in the YES case we have the additional property that $A$ has full rank. The hardness of $(1 - 1/n^c)$-Approximate linear search problem with the same properties follows from the decision to search reduction from Lemma~\ref{lem:dec_to_search-w}.
\end{proof}

We now prove Corollary~\ref{cor:search-amp_w}:

\begin{proof}[Proof of Corollary~\ref{cor:search-amp_w}]
    Note that the input size is $\tO(n^2)$ and hence $a \geq 2$. The corollary directly follows from noting that Lemma~\ref{lem:amp-w} applied for $\eps(n) = 1/n$ and $\delta(n) = 1/n$ reduces $(1/n)$-Approximate Linear Search problem to $(1-1/n)$-Approximate Linear Search problem in time $\tO(n^2) = \tO(n^a)$.
\end{proof}

\subsection{Sparse Matrices}
Starting from the WordRAM version of Conjecture~\ref{conj:rk-sparse} and using Lemma~\ref{lem:red-main_w} we can establish that there does not exist a $\tilde{o}(n^2)$ algorithm for $(1 - 1/\poly(n))$-Approximate linear decision problem on sparse matrices in the WordRAM model. By the decision to search reduction in Lemma~\ref{lem:dec_to_search-w} we get that there does not exist a $\tilde{o}(n^2)$ algorithm for $(1 - 1/\poly(n))$-Approximate linear decision search on sparse matrices in the WordRAM model. Note though that this hardness is trivial to obtain since there exist sparse linear systems such that every $(1 - 1/\poly(n))$-approximate solution to the system requires $\Omega(n^2)$ bits to represent.

On the algorithmic side no improvement over the dense case algorithmic runtime of $O(n^{\omega})$ was known until the recent result of Peng and Vempala~\cite{PengV21} who gave an asymptotically faster algorithm for the $1/\poly(n)$-approximate linear search problem.

\section{Different notions of approximation}\label{sec:approx-discuss}

Our notion of ``Linear Search $\epsilon$-Approximation problem'' is easier than how the notion of approximately solving is usually defined~\cite{SpielmanT14,KyngZ17}:

\begin{definition}[Linear System Approximation Problem]
Given a satisfiable linear system $(A\in \R^{m \times n} , b \in \R^{m})$ find an $x \in \R^n$ such that: 
$$\norm{Ax-\Pi_A(b)}_2 \leq \epsilon\norm{\Pi_A(b)}_2,$$
where $\Pi_A(b)$ denotes the projection of the vector $b$ onto the column space of $A$.
\end{definition}

The above notion considers general linear systems rather then just satisfiable ones. For satisfiable linear systems their notion is equivalent to ours as in that case $\Pi_a(b) = b$. As we give conditional hardness for ``Linear Search $\epsilon$-Approximation'' this directly implies hardness for ``Linear System Approximation Problem'' (referred to as LSA in Kyng and Zhang~\cite{KyngZ17}). On the other hand, our $\tOm(n^\omega)$ hardness is tight even for finding a $x$ such that $\norm{Ax-\Pi_A(b)}_2 = 0$ as we can do that in time $O(n^\omega)$~\cite{IbarraMH82} by computing the pseudoinverse.

\section{Equivalent Characterization of Our Conjectures}\label{sec:eq}

We will now prove that the Linear decision problem, the Rank-finding problem and the full rank problem are all equivalent. Thus our conjectures about hardness of Rank-Finding can be interpreted as hardness of any of these problems.

\begin{lemma}\label{lem:eq}
There exists randomized turing reductions (to upto $\plog(n)$ instances) in RealRAM which work w.h.p. between the following problems:
\begin{enumerate}
    \item Linear Decision Problem: Given $(A, b)$ with $A \in \mathbb{R}^{O(n) \times O(n)}$ with $\nnz(A) = z+\tO(n)$ does there exist a $x$ such that $Ax = b$.
    \item Rank-Finding Problem: Given $B \in \mathbb{R}^{O(n) \times O(n)}$ with $\nnz(B) = z+\tO(n)$ find $\rk(B)$.
    \item Full-Rank Problem: Given $C \in \mathbb{R}^{n' \times n'}$ with $n' = O(n), \nnz(C) = z + \tO(n)$, is $\rk(C) = n'$? or equivalently is the determinant of $C$ non-zero.
\end{enumerate}
\end{lemma}

\begin{proof}
We have already reduced the Rank-Finding problem to the Full Rank problem in Lemma~\ref{lem:rank} and Full-Rank Problem to Linear Decision Problem in Lemma~\ref{lem:fc_to_lda} while preserving sparsity in both the reductions.

We now reduce Linear Decision Problem to the Rank-Finding Problem. 
Suppose we want to check satisfiability of the linear system $(A, b)$. Consider the matrix $A' = [A ~~ b~]$, i.e. the matrix formed by appending the column vector $b$ to the matrix $A$. The system $(A, b)$ has a solution if and only if $rk(A) = rk(A')$, so it suffices to  find the rank of the matrices $A,A'$. Hence we have reduced the linear system decision problem to the Rank-Finding Problem. Also note that $\nnz(A') = z+m = z+O(n)$.
\end{proof}

In fact as Theorem~\ref{thm:red-main} reduces Rank-Finding problem to $(1 - \delta(n))$-Approximate linear decision problem for all $\delta(n) > 0$ which is a subcase of linear decision problem. Hence we get that $(1 - \delta(n))$-Approximate linear decision problem for all $\delta(n) > 0$ is also equivalent to the all three problems in Lemma~\ref{lem:eq}.

A similar equivalence exists in WordRAM:

\begin{lemma}\label{lem:eq_w}
There exists randomized turing reductions (to upto $\plog(n)$ instances) in RealRAM which work w.h.p. between the following problems:
\begin{enumerate}
    \item Linear Decision Problem: Given a linear system $(A, b)$ with $A \in \mathbb{Z}^{O(n) \times O(n)}, \nnz(A) = z+\tO(n), \bc(A), \bc(b) = O(\log n)$ does there exist a $x$ such that $Ax = b$.
    \item Rank-Finding Problem: Given $B \in \mathbb{Z}^{O(n) \times O(n)}$ with $\nnz(B) = z+\tO(n), \bc(b) = O(\log n)$ find $\rk(B)$.
    \item Full-Rank Problem: Given $C \in \mathbb{R}^{n' \times n'}$ with $n' = O(n), \nnz(C) = z+\tO(n), \bc(C) = O(\log n)$, is $\rk(C) = n'$? or equivalently is the determinant of $C$ non-zero.
\end{enumerate}
\end{lemma}

\begin{proof}
We have already reduced the Rank-Finding problem to the Full Rank problem in Lemma~\ref{lem:rank_w} and Full-Rank Problem to Linear Decision Problem in Lemma~\ref{lem:fc_to_lda_w} while preserving sparsity in both the reductions.

We now reduce Linear Decision Problem to the Rank-Finding Problem. 
Suppose we want to check satisfiability of the linear system $(A, b)$. Consider the matrix $A' = [A ~~ b~]$, i.e. the matrix formed by appending the column vector $b$ to the matrix $A$. The system $(A, b)$ has a solution if and only if $rk(A) = rk(A')$, so it suffices to  find the rank of the matrices $A,A'$. Hence we have reduced the linear system decision problem to the Rank-Finding Problem. Also note that $\nnz(A') = z+m = z+O(n)$ and $\bc(A') = \max(\bc(A), \bc(b) = O(\log n)$. 
\end{proof}

\newpage

\bibliography{references}

\end{document}